\documentclass[11pt]{article}
\usepackage{tikz}
\usepackage[colorinlistoftodos,textsize=footnotesize]{todonotes}
\usetikzlibrary{matrix,arrows,shapes}
\usepackage{pgf}
\usepackage{cancel}
\usepackage{url}
\usepackage{longtable}
\usepackage{enumerate}
\usepackage{amsthm}
\usepackage{amssymb}
\usepackage{amsmath}
\usepackage{amscd}
\usepackage{eucal}
\usepackage{mathrsfs}
\usepackage{upgreek}
\usepackage{mathtools}
\usepackage{dsfont}
\usepackage{yfonts}
\usepackage[T1]{fontenc}
\usepackage{bbm}
\usepackage{geometry}
\usepackage{graphbox}
\usepackage{array}
\usepackage[retainorgcmds]{IEEEtrantools}
\usepackage{booktabs}
\usepackage{enumerate}
\usepackage{feynmf}

\newcommand{\astunder}{\underline{*}}
\newcommand{\astover}{\overline{*}}
\usetikzlibrary{trees}
\usetikzlibrary{decorations.pathmorphing}
\usetikzlibrary{decorations.markings}
\tikzset{
	photon/.style={decorate, decoration={snake}, draw=red},
	electron/.style={draw=blue, postaction={decorate},
		decoration={markings,mark=at position .55 with {\arrow[draw=blue]{>}}}},
	gluon/.style={decorate, draw=magenta,
		decoration={coil,amplitude=4pt, segment length=5pt}},
	sderiv/.style={postaction={decorate},
		decoration={markings,mark=at position .3 with {\arrow{>}}}},
	tderiv/.style={postaction={decorate},
		decoration={markings,mark=at position .7 with {\arrow{<}}}},
	stderiv/.style={postaction={decorate},
		decoration={markings,mark=at position .7 with {\arrow{<}},mark=at position .3 with {\arrow{>}}}}
}
\usepackage{scalerel}
\usepackage{stackengine}

\definecolor{see}{RGB}{67,75,179}
\definecolor{darksee}{RGB}{42,44,148}
\definecolor{honey}{RGB}{232,180,129}
\definecolor{lighthoney}{RGB}{255,254,220}
\definecolor{citecol}{rgb}{0.5,0,0} 
\definecolor{blue1}{RGB}{130,150,209}
\usepackage{graphicx,stackengine,scalerel}

\definecolor{blue1}{RGB}{130,150,209}
\definecolor{blue2}{RGB}{42,60,122}
\usepackage{soul}
\usepackage[colorinlistoftodos,textsize=footnotesize]{todonotes}

\setstcolor{red}
\definecolor{see}{RGB}{67,75,179}
\usepackage{hyperref}
\hypersetup
{
	bookmarks=true,
	bookmarksopen=true,
	pdfmenubar=true, 
	pdfhighlight=/O, 
	colorlinks=true, 
	pdfpagemode=None, 
	pdfpagelayout=SinglePage, 
	pdffitwindow=true, 
	linkcolor=see, 
	citecolor=red!75!darkgray, 
	urlcolor=see 
}
\newlength{\bibitemsep}
\newlength{\bibparskip}
\let\oldthebibliography\thebibliography
\renewcommand\thebibliography[1]{%
	\oldthebibliography{#1}%
	\setlength{\parskip}{\bibitemsep}%
	\setlength{\itemsep}{\bibparskip}%
}
\newcommand{\no}[1]{\mathop{\mathopen: {#1} \mathclose:}}
\newcommand{\CE}{\mathcal{CE}}  
\newcommand{\E}{\mathcal{E}}

\newcommand{\V}{\mathcal{V}}

\newcommand{\fA}{\mathfrak{A}}

\newcommand{\F}{\mathcal{F}}

\newcommand{\BV}{\mathcal{BV}}

\newcommand{\frakg}{\mathfrak{g}}
\newcommand{\frakk}{\mathfrak{k}}

\newcommand{\euL}{\mathscr{L}}

\newcommand{\euV}{\mathscr{V}}  

\newcommand{\Gcal}{\mathcal{G}}  

\newcommand{\BVcal}{\mathcal {BV}}

\newcommand{\CEcal}{\mathcal {CE}}

\newcommand{\Dcal}{\mathcal{D}}
\newcommand{\Ecal}{\mathcal{E}} 

\newcommand{\Fcal}{\mathcal{F}} 
\newcommand{\Acal}{\mathcal{A}} 

\newcommand{\Ncal}{\mathcal{N}}
\newcommand{\Mcal}{\mathcal{M}}
\newcommand{\Ocal}{\mathcal{O}}
\newcommand{\Scal}{\mathcal{S}}
\newcommand{\Pcal}{\mathcal{P}}

\newcommand{\Tcal}{\mathcal{T}}
\newcommand{\Vcal}{\mathcal{V}}

\newcommand{\Ci}{\mathcal{C}^\infty} 
\newcommand{\WF}{\mathrm{WF}}         
\newcommand{\id}{\mathrm{id}}               
\newcommand{\Diff}{\mathrm{Diff}}        
\DeclareMathOperator{\im}{\mathrm{Im}}             
\DeclareMathOperator{\tr}{\mathrm{tr}}                 

\newcommand{\loc}{\mathrm{loc}}
\newcommand{\inv}{\mathrm{inv}}
\newcommand{\reg}{\mathrm{reg}}
\newcommand{\ren}{\mathrm{r}}

\newcommand{\pg}{\mathrm{pg}}

\newcommand{\af}{\mathrm{af}}
\newcommand{\ta}{\mathrm{ta}}
\newcommand{\gh}{\mathrm{gh}}

\newcommand{\mc}{{\mu\mathrm{c}}}
\newcommand{\ml}{\mathrm{ml}}
\newcommand{\ex}{\mathrm{ext}}

\newcommand{\YM}{\rm\sst YM}
\newcommand{\NN}{\mathbb{N}}          
\newcommand{\RR}{\mathbb{R}}           
\newcommand{\CC}{\mathbb{C}}           

\newcommand{\al}{\alpha}

\newcommand{\De}{\Delta}

\newcommand{\la}{\lambda}

\newcommand{\ph}{\varphi}

\newcommand{\T}{\cdot_{{}^\Tcal}}

\newcommand{\TR}{\cdot_{{}^{\TTR}}}

\newcommand{\TT}{\Tcal}
\newcommand{\TTR}{{\Tcal_\ren}}

\newcommand{\Poi}[2]{\left\lfloor#1,#2\right\rfloor}

\newcommand{\qme}{{\textsc{qme}}}

\newcommand{\cme}{{\textsc{cme}}}
\newcommand{\mwi}{{\textsc{mwi}}}

\newcommand{\sst}[1]{\scriptscriptstyle{#1}}  
\newcommand{\minus}{\sst{-1}}   
\newcommand{\1}{\mathds{1}}                         
\newcommand{\pa}{\partial}                              
\newcommand{\be}{\begin{equation}}
\newcommand{\ee}{\end{equation}}

\newcommand{\Lap}{\bigtriangleup}

\newcommand{\form}{\stackrel{\mathrm{formal}}{=}}

\DeclareMathOperator{\supp}{supp}      

\newcommand{\Pei}[2]{\lfloor #1, #2 \rfloor}
\newcommand{\skal}[2]{\left< #1 , #2 \right>}

\DeclareMathOperator{\Aut}{Aut}

\newcommand{\dgr}{{\sst\ddagger}}
\newcommand{\Jac}{\mathrm{Jac}}
\theoremstyle{plain}
\newtheorem{thm}{Theorem}[section]
\newtheorem{prop}[thm]{Proposition}

\newtheorem{cor}[thm]{Corollary}
\newtheorem{lemma}[thm]{Lemma}

\theoremstyle{definition}
\newtheorem{df}[thm]{Definition}
\theoremstyle{remark}
\newtheorem{rem}[thm]{Remark}

\newtheorem{exa}[thm]{Example}

	\title{BV quantization in perturbative  algebraic QFT: 
		\\{\Large \textit{Fundamental concepts and perspectives}}}
	\date{\today}
	\author{Kasia Rejzner}
	
\begin{document}	
	\maketitle	
	\begin{abstract}
		This paper is mainly based on the talk I presented at the meeting \textit{The Philosophy and Physics of Noether's Theorems} that took place 5-6 October 2018, but it also contains some original results that were inspired by discussions with mathematicians, physicists and philosophers about the problem of understanding the intrinsic meaning of gauge invariance. In this work, I argue that following the principles of \textit{locality}, \textit{deformation} and \textit{homology}, one naturally ends up using the Batalin-Vilkovisky (BV) formalism in quantizing gauge theories. 
		
		I start with the gentle introduction into the BV framework and then I proceed to some new results and more speculative deliberations. In the classical theory, I present a new perspective on the classical BV operator, using the notion of M{\o}ller maps. In the quantum theory, I present some loose ideas on the formulation of anomalous master Ward identity in the framework proposed recently by Buchholz and Fredenhagen, based on local S-matrices.
	\end{abstract}

\section{Introduction}
The rigorous formulation of gauge theories is one of the big open problems in modern mathematical physics. Its importance for mathematics is evidenced by the fact that it is one of the famous Millennium Problems formulated by the Clay Mathematics Institute. Understanding the proper meaning of concepts such as \textit{symmetry} and \textit{duality} in the context of gauge theories (both classical and quantum) is also an important issue from the point of view of philosophy of science. Gauge theories are therefore a perfect ground for physicists, mathematicians and philosophers to work together and develop new concepts to tackle the fundamental questions.

I think one can say that the first fundamental result concerning symmetries in mathematical physics was obtained by Emmy Noether, who made the connection between symmetries and conservation laws \cite{Noether} (see also the book \cite{IKS} of Yvette Kosmann-Schwarzbach for a review and historical perspective). The conserved currents and charges that one obtains using Noether's theorems are not only crucial in classical field theory, but also play important role in quantization.

In quantum field theory (QFT), in order to quantize models with local symmetries (including gauge theories and gravity), it is convenient to extend the original system (given in terms of a Lagrangian) by adding auxiliary fields (ghosts, antighosts, etc.). The extended system enjoys  a rigid symmetry called BRST symmetry (named after its proposers Becchi, Rouet, Stora \cite{r00149,r00110} and independently Tyutin \cite{r01653}) that is an extension of the original local symmetry. For this symmetry one can then construct the corresponding Noether charge and quantize this charge to  implement the symmetry at the quantum level. This idea is sometimes referred to as the \textit{quantum Noether method}
\cite{r00086}.

Soon after the BRST, a generalization of their method, called \textit{BV formalism}, has been proposed by Batalin and Vilkovisky \cite{BV81}. Later on, it was recognized that this framework for quantization of gauge theories has deep connections to homological algebra, graded geometry and modern concepts of derived geometry. There is plenty of literature on the subject, so it is hard to give a really comprehensive list. Here I just mention a few titles: \cite{HT} is a standard reference, \cite{r00076} gives a really comprehensive treatment of Yang-Mills theories, \cite{GPS94} provides a physicist-oriented review of the formalism for beginners and the recent review \cite{CM19} also includes treatment of theories with boundary. For connections to homolofical algebra, homotopy theory and derived geometry see for example \cite{r01658, r01695, PTVV, Calaque18}

Nowadays, the BV formalism plays a central role in two mathematically rigorous approaches to perturbative QFT: 
\begin{itemize}
		  \setlength\itemsep{-0.2em}
	\item \textbf{Perturbative algebraic quantum field theory (pAQFT)}, which is a mathematically rigorous framework for perturbative QFT that combines the main ideas of the axiomatic framework of algebraic QFT \cite{HK,Haag,Araki99} with perturbative methods for constructing models. The main contributions focusing on the scalar field are:  \cite{DF,DF02,DF04,DF05,DFloop,BreDue,Boas,BoasDue,BDF}. Abelian gauge theories were treated in \cite{DFqed}, while the Yang-Mills theories are the subject of \cite{H}. The full incorporation of the BV framework has been done in \cite{FR,FR3,Rej11b}.	
	\item \textbf{Factorisation algebras} approach by Costello and Gwilliam \cite{Cos,CoGw}.
\end{itemize}
For the relations between the two approaches, see \cite{GR20} and a later work \cite{BPS19}.

In this paper, I would like to focus on the conceptual understanding of the BV formalism from the point of view of pAQFT. I will focus on the definition of the classical and quantum BV operator. After reviewing the basics of the formalism, I will propose a new approach to defining the classical BV operator and show how this naturally leads to quantization in agreement with \cite{FR3}. In the second part, I will present some new, perhaps speculative, ideas on how one could incorporate the BV formalism into a non-perturbative framework, along the lines of \cite{BF19}.

I argue that the BV formalism in pAQFT is the result of applying the following principles to quantization of gauge theories:
\begin{itemize}
	  \setlength\itemsep{-0.2em}
	  	\item \textit{Locality}: imposing and preserving locality is crucial for renormalization and it allows one to identify the space of observables of interest. More on locality in pAQFT and its generalizations can be found in the recent review \cite{Rej19}.
	  \item \textit{Deformation}: in pAQFT, quantum models are built from the classical models by means of deformation quantization. The advantage is that one stays with the same space of observables and what changes is the product. More generally, in the \cite{FR3} version of the BV formalism, one also deforms differential graded algebras.
	  \item \textit{Homology}: one avoids building explicit quotients or spaces of orbits under a Lie algebra (or group) action. Instead, one uses homological algebra to construct their derived versions. This idea has been very fruitful in mathematics and turns out to be also very successful in understanding the foundations of QFT. 
\end{itemize}
Through the combination of these three principles one arrives at a formulation that is both mathematically rigorous and physically meaningful. I will now outline the basics of the resulting framework.

\section{Classical field theory}
I start with classical theory. I will follow the principles of \textit{locality} and \textit{homology} to construct models, which will then be easy to quantize using \textit{deformation}.
\subsection{Kinematical structure}
Let  $M$  be an oriented, time-oriented globally hyperbolic spacetime (the latter effectively means that it has a Cauchy surface). The conceptual results presented here hold for general theories with local gauge invariance, but as a running example I will work with the self-interacting Yang-Mills theory. 

Let $\Ecal$ denote the configuration space of the theory, which in this paper is always understood as the space of smooth sections of some vector bundle $E\xrightarrow{\pi} M$ over $M$. By choosing $\Ecal$, we decide what kind of objects the theory describes (e.g. scalar fields, tensor fields). We will need some notation:
\begin{itemize}
		  \setlength\itemsep{-0.2em}
	\item $\Ecal_c$ denotes the space of smooth compactly supported sections of $E$.
	\item  $\Ecal'$,  $\Ecal'_c$ denote complexifications of topological duals of  $\Ecal$ and  $\Ecal_c$ respectively (both equipped with the strong topology).
	\item $\Ecal^*$ denotes the space of smooth sections of the  dual bundle $E^*$. 
	\item $\Ecal^!$ denotes the complexification of the space of sections of $E^*$ tensored with the bundle of densities over $M$. By a slight abuse of notation, I use the symbol $\Ecal^!(M^n)$ for the complexified space of sections of the $n$-fold exterior tensor product of this bundle, seen as a vector bundle over $M^n$.
\end{itemize}
In this paper, I will typically denote the elements of $\Ecal$ by $\ph$, even if they carry indices. This makes the notation simpler. I will invoke the indices again later, when it becomes necessary.
\begin{exa}
	I will focus on two running examples:
	\begin{itemize}
			  \setlength\itemsep{-0.2em}
	\item \textit{Scalar field}. Here the configuration space is just $\Ecal=\Ci(M,\RR)$. 
	 \item \textit{Yang-Mills theories}. I consider
	$G$, a semisimple compact Lie group and $\frakk$ its Lie algebra. For simplicity, let's take the trivial bundle\footnote{Non-trivial bundles can also be treated, but in this paper I want to focus on the perturbative treatment of Yang-Mills theory, so restricting to trivial bundles is sufficient. For the non-perturbative treatment of the classical configuration space see \cite{BSS17}.} $P=M\times G$ over $M$ and define the off-shell configuration space of the Yang-Mills theory as $\Ecal=\Omega^1(M,\frakk)$.
	\end{itemize}
\end{exa}
We model classical observables as functionals on $\Ecal$. To make this mathematically precise, let's equip  $\Ecal$ with its natural Frech\'et topology and consider the space of (Bastiani) smooth functionals $\Ci(\Ecal,\CC)$ (see \cite{Bas64} for the details on the appropriate notion of smoothness). These will be our observables. To understand the physical motivation, note that classically, an observable assigns to a given field configuration a number, which corresponds to the value of the measurement of that observable (e.g. energy density at a given point in spacetime). The smoothness requirement assures that all the algebraic structures that we want to introduce on these observables are well defined.

The next important notion is that of a \textit{spacetime support} of a functional. It encodes localization properties of observables. Another crucial property is \textit{additivity}.
\begin{df}
	The spacetime support of a functional is defined by
	\begin{align}\label{support}
	\supp\, F=\{ & x\in M|\forall \text{ neighbourhoods }U\text{ of }x\ \exists \ph_1,\ph_2\in\E, \supp\, \ph_2\subset U 
	\\ & \text{ such that }F(\ph_1+\ph_2)\not= F(\ph_1)\}\ .\nonumber
	\end{align}
\end{df}
\begin{df}
	A functional $F$ is called \textit{additive} if
	\be\label{L:add}
	F(\ph+\chi+\psi)=F(\ph+\chi)-F(\chi)+F(\chi+\psi)\,,
	\ee
	for $\ph+\chi+\psi\in\Ecal$ and $\supp\,\ph\cap\supp\,\psi=\emptyset$. 
\end{df}
On one hand this property is as a weaker version of linearity and on the other hand it encodes \textit{locality}.  Recall that, in physics, a functional is called \textit{local}  if it can be written in the form:
\[
F(\ph)=\int\limits_M \omega(j^k_x(\ph))\,d\mu(x)\,,
\]
where $\omega$ is a function on the jet bundle over $M$ and $j^k_x(\ph)=(x,\ph(x),\pa \ph(x),\dots)$, with derivatives up to order $k$, is the $k$-th jet of $\ph$ at the point $x$. It was shown in \cite{BDGR} (based on ideas presented in \cite{BFR}) that local functionals can be characterised as smooth functionals that obey \eqref{L:add} and have smooth first derivatives. More about additivity and its  generalizations can be found in \cite{Rej19}.

The space of compactly supported smooth local functions on $\Ecal$ is denoted by $\Fcal_\loc$. The algebraic completion of $\Fcal_\loc$ with respect to the pointwise product
\be\label{prod}
F\cdot G(\ph)=F(\ph)G(\ph) \,,
\ee
is the commutative algebra $\Fcal$ of \textit{multilocal functionals}.

We also introduce \textit{regular functionals}. We say that $F\in \Fcal_{\reg}$ if all the derivatives $F^{(n)}(\ph)$ are smooth, i.e. for all $\ph\in\Ecal$, $n\in \NN$ we have
\[
F^{(n)}(\ph)\in\Ecal^!(M^n)\,.
\]

\subsection{Dynamics and symmetries}
To introduce dynamics, I use a generalization of the Lagrangian formalism, following \cite{BDF}. Ideally, we would like to be able to derive the equations of motion and symmetries from the action principle. The potential difficulty here is that the manifolds we are working with are non-compact so the integral of a Lagrangian density, for example 
$\frac{1}{2}(\nabla^\mu\ph \nabla_\mu\ph -m^2 \ph^2)$,
over the whole manifold $M$ does not converge, if $\ph$ is not compactly supported. One could be tempted to restrict attention to compactly supported configurations, but this does not work either, since the equations of motion we want to consider do not have non-trivial compactly supported solutions! To get around this obstruction, we smear the Lagrangian density with a cutoff function $f\in\Dcal\doteq\Ci_c(M,\RR)$ and define all the relevant objects (e.g. the Euler-Lagrange derivative) in a way which is independent of $f$. To make this more systematic, let's introduce the notion of a \textit{generalized Lagrangian}.
\begin{df}\label{Lagr}
	A \textit{generalized Lagrangian} on a fixed spacetime $M$ is a map $L:\Dcal\rightarrow\Fcal_{\loc}$ such that
	\begin{enumerate}[i)]
		\item $L(f+g+h)=L(f+g)-L(g)+L(g+h)$ for $f,g,h\in\Dcal$ with $\supp\,f\cap\supp\,h=\varnothing$ ({\bf Additivity}).
		\item $\supp(L(f))\subseteq \supp(f)$ ({\bf Support}).
		\item Let $\Gcal$ be the isometry group of the spacetime $M$ (for Minkowski spacetime we set $\Gcal$ to be the proper orthochronous Poincar\'e group $\Pcal^\uparrow_+$.). We require that $L(f)(g^*\ph)=L(g_*f)(\ph)$ for every $g\in\Gcal$ ({\bf Covariance}).
	\end{enumerate}
 I denote the space of all generalized Lagrangians by $\euL$. 
\end{df}
I also assume that the Lagrangians satisfy $L^*=L$ with respect to the involution $*$, which is for now just the complex conjugation, but when we get to graded geometry, the involution will also swap the order of factors.

Now it's time to get rid of the dependence on $f$.
\begin{df}[\cite{BDF}] 
	Actions $S(L)$ are defined as  equivalence classes of Lagrangians, where two Lagrangians $L_1,L_2$ are called equivalent $L_1\sim L_2$  if
	\be\label{equ}
	\supp (L_{1}-L_{2})(f)\subset\supp\, df\,, 
	\ee
	for all $f\in\Dcal$. 
\end{df}
The idea is, essentially to identify generalized Lagrangians whose defining Lagrange densities differ by a total derivative. From here on, I will use the notation $S$ rather than $S(L)$ and all objects and constructions that do not depend on the choice of representative in the equivalence class $S$ will be labelled with $S$ rather than $L$.
\begin{exa}
	The generalized Lagrangian of the free scalar field is
\[
L_0(f)[\ph]=\frac{1}{2}\int_M (\nabla^\mu\ph \nabla_\mu\ph -m^2 \ph^2)f d\mu_g\,.
\]
For the Yang-Mills theory, we have
\[
L_{\YM}(f)[A]=-\frac{1}{2}\int_M f\,\tr(F \wedge * F)\,,
\]
where $F=dA+\frac{i\lambda}{2}[A,A]$, $A\in\Ecal$,  $\lambda$ is the coupling constant, $*$ is the Hodge operator and $\tr$ is the trace in the adjoint representation, given by the Killing-Cartan metric $\kappa$.
\end{exa} 

Following \cite{BF19}, I introduce some further notation. 
\begin{df}\label{df:delta:L}
	Let $L\in \euL$, $\ph\in\Ecal$. Define a functional $\delta L:\Dcal\times\Ecal\rightarrow \RR$ by
	\[
	\delta L(\psi)[\ph]\doteq L(f)[\ph+\psi]-L(f)[\ph]\,, 
	\]
	where $\ph\in\Ecal$, $\psi\in \Dcal$ and $f\equiv 1$ on $\supp \psi$ (the map $\delta L(\psi)[\ph]$ thus defined does not depend on the particular choice of $f$).
\end{df}
The above definition can be turned into a difference quotient and we can use it to introduce the \textit{Euler-Lagrange derivative} of $S$.
The equations of motion are understood in the sense of \cite{BDF}. Concretely, the Euler-Lagrange derivative of $S$ is a 1-form on $\Ecal$, i.e. a map $dS:\Ecal\to\Ecal_c'$ defined by
\be\label{ELd}
\left<dS(\ph),\psi\right>\doteq \lim_{t\rightarrow 0}
\tfrac{1}{t}\delta L(t\psi)[\ph]=\int\frac{\delta L(f)}{\delta \ph (x)} \psi(x) \,,
\ee
with $\psi\in\Ecal_c$ and $f\equiv 1$ on $\supp \psi$. Here $\frac{\delta L(f)}{\delta \ph}$ is understood as an element of $\Ecal^!\subset \Ecal_c'$. The field equation is now the following condition on $\ph$:
\be
dS(\ph)\equiv 0\,,\label{eom}
\ee
	so geometrically, the solution space is the \textit{zero locus of the 1-form $dS$}. Note that $dS$ lives on $\Ecal$, rather than $M$! Let $\Ecal_S\subset \Ecal$ denote the space of solutions to \eqref{eom}. We are interested in the space $\Fcal_S$, of functionals on $\Ecal_S$. We will call them \textit{on-shell} functionals. 
\begin{exa} Examples of equations of motion:	
	\begin{itemize}
			  \setlength\itemsep{-0.2em}
\item	\textbf{Free scalar field}:
	$dS_0(\ph)=-(\Box+m^2)\ph$,
	where $\Box$ is the wave operator (d'Alembertian).
\item \textbf{Yang-Mills theory}:	
$dS_{\YM}(A)=D_A\!*\!F$,
where $D_A$ is the covariant derivative induced by the connection $A$. 
\end{itemize}
\end{exa}
\begin{rem}
For systems with several fields (or components), I will use the notation $\frac{\delta S}{\delta \ph^{\alpha}}$ for  $\frac{\delta L(f)}{\delta \ph^{\alpha}}$ evaluated at $f\equiv 1$ and treated as a component of the form $dS$. Here $\alpha$ runs from $1$ to $D$, where $D$ is the number of degrees of freedom of the system (for the scalar field it is 1, for the pure Yang-Mills it is equal to 4 times the dimension of $\mathfrak{k}$).
\end{rem}

Next, I  discuss \textit{symmetries}. These are directions in the configuration space $\Ecal$, around a given point, along which the action is constant. Geometrically, these are vector fields $X$ on $\Ecal$ such that
\[
\partial_X S\equiv 0\,,
\]
where
\[
\partial_X S\doteq \int\frac{\delta L(f)}{\delta \ph(x)} X(x) \,,\quad f\equiv 1\ \textrm{on}\ \supp X\,,
\]
and $X\in \Gamma(T\Ecal)$ is identified with a map from $\Ecal$ to $\Ecal_c^{\sst \CC}$. Note that $\partial_X S$ is just the insertion of the 1-form $dS$ into a vector field $X$.

Formally, we write
\[
X=\int  X(x) \frac{\delta}{\delta \ph(x)} \,,
\]
and identify the basis on the fiber $T_\ph\Ecal$ as the antifields $\frac{\delta}{\delta \ph(x)}\equiv \ph^\ddagger(x)$.

I will now focus on the situation, where any local symmetry of the system can be expressed as 
$$X=\omega\rho(\xi) + I\,,$$
 where $I$ is a symmetry that vanishes identically on $\Ecal_S$, $\omega$ is a local function from $\Ecal$ to $\Dcal$ (multiplication with an element of $\Gamma(T\Ecal)$ is defined fiberwise) and $\rho:\frakg_c\rightarrow \Gamma (T\Ecal)$ is a Lie-algebra morphism, arising from a given local action $\sigma$ of some Lie algebra $\frakg_c$ on $\Ecal$, by means of
 \[
 \rho(\xi)F[\ph]:=\left<F^{(1)}(\ph),\sigma(\xi)\ph\right>\equiv\int_M \frac{\delta F}{\delta \ph(x)}\sigma(\xi)\ph(x)\,.
 \]
 We assume $\frakg_c$ to be the space of smooth compactly supported (hence the subscript $c$) sections of some vector bundle over $M$ and  the action $\sigma$ on $\Ecal$ to be local.
\begin{exa}
For Yang-Mills theory, compactly supported local symmetries are given in terms of the Lie algebra $\frakg_c=\Gamma_c(M,\mathfrak{k})$ and the local action $\sigma$ is given by 
\[
\sigma(\xi)A:=d\xi+[A,\xi]=D_A\xi\,,\quad \xi\in\frakg_c\,.
\]
\end{exa}
The presence of local symmetries implies that the equations of motion have redundancies, or, in other words, that $\Ecal_S$, the zero locus of $dS$ consists of \textit{orbits} of the action $\sigma$ of $\frakg$ on $\Ecal$.  To see this explicitly, note that,  since $X$ is assumed to be local and compactly supported, it can be expressed in terms of some differential operator
\[
X^\alpha(x)[\ph]=Q^{\alpha}_{\ \beta}(\ph) \ph^\beta(x)= a(x)[\ph]\ph(x)+b^\mu(x)[\ph]\nabla_\mu\ph(x)+\dots\,,
\]
so the condition for $X$ to be a symmetry can be expressed as
\be\label{noether}
0=\int \frac{\delta S}{\delta \ph^\alpha(x)}X^\alpha(x) d\mu(x)=\int \ph^\beta(Q^{\alpha}_{\ \beta})^*\frac{\delta S}{\delta \ph^\alpha} d\mu\,,
\ee
where $*$ denotes the formal adjoint of a differential operator (obtained using integration by parts).  This is the second Noether theorem and it leads to the conclusion that $\frac{\delta S}{\delta \ph^\alpha}$, the equations of motion of the system, are not all independent. We will come back to this point in section \ref{Kc}. More on the relation between Noether's second theorem and the BV formalism can be found in \cite{FLS03}.

Ultimately, we are interested in functionals on the solution space $\Ecal_S$ that are \textit{invariant} under the action $\rho$ of the symmetries. We will denote this space by $\Fcal_{S}^{\inv}$.

\subsection{Homological interpretation}
Our goal is to characterize the space $\Fcal_{S}^{\inv}$ of symmetry-invariant on-shell functionals in a way that will facilitate quantization. Remember that our aim is not just to construct the classical theory, but rather to use it as a first step towards quantization.

 The conclusion from Noether's second theorem (which we have now re-phrased in a slightly different language, following) is that, in the presence of local symmetries, \textit{equations of motion have redundancies}, so the Cauchy problem is not well posed in such systems and one is tempted to remove the redundancy by taking the quotient by the action $\rho$ of infinitesimal symmetries.  However, following the guiding idea of \textit{homology}, instead of forming a quotient, we can go to a larger space where the equations of motion are better behaved and we can keep track of relations between equivalent solutions. 
 
 From the point of view of deformation quantization that we want to perform in the end, it is more convenient to work with vector spaces and encode the information about symmetries and equations of motion in maps between these vector spaces. Remarkably, the kind of algebra that one uses in this construction also started with Emmy Noether! 
 

\subsubsection{Koszul complex}\label{Kc}
We start with finding homological interpretation for $\Fcal_S$. The idea is to describe it as the quotient $\Fcal_{S}=\Fcal/\Fcal_0$, where $\Fcal_0$ is the space of functionals that vanish on $\Ecal_S$ (redundancy removal by quotienting). How do find such functionals? There is a nice geometrical way to do it.
 Recall that the space of solutions $\Ecal_S$ is the space on which the one form $dS$ vanishes. If we take a vector field $X$, then $\iota_{dS}X$ (the insertion of a one-form into a vector field) vanishes identically on $\Ecal_S$. If we assume $X$ to have appropriate locality properties, then  $\iota_{dS}X\in \Fcal_{0}$. Let's define $\V$ to be the space of multilocal vector fields on $\Ecal$ (for the precise definition, see \cite{Book}) and introduce a map $\delta_S:\V\rightarrow \Fcal$ by setting
 \[
 \delta_S(X):=-\iota_{dS}(X)\,.
 \]
Clearly, the image of $\delta_S$ is contained in $\Fcal_{0}$. One could ask the question whether it is in fact all of $\Fcal_{0}$. This is less obvious and depends on the system. It can be shown that $\Fcal_{0}$ is equal to the image of $\delta_S$, provided $S$ satisfies certain \textit{regularity conditions} (see e.g. \cite{r00075,HT}). One requires that the equations of motion of the system can be split into independent ones 
\be\label{indep}
\frac{\delta S}{\delta \ph^{\alpha}}(\ph)=0\,,\quad \alpha=1,\dots, N
\ee
and $D-N$ of dependent ones (the relations follow from Noether's second theorem \eqref{noether}), so that the full system of equations $dS(\ph)=0$ is fully equivalent to \eqref{indep}. Note that $\frac{\delta S}{\delta \ph^{\alpha}}$ are local functions from $\Ecal$ to $\Ecal$ (i.e. depend only on $\ph$ and its derivatives at a point), so can be seen as functions on the jet space (a $k$-jet of $\ph$ at point $x$ is essentially given by $(\ph(x),\partial\ph(x),\dots)$ with derivatives up to order $k$). Assume that $\frac{\delta S}{\delta \ph^{\alpha}}\equiv S_{\alpha}$, $\alpha=1,\dots,N$ can be used (at least locally) as the first $N$ coordinates on the jet space. It is crucial that functionals in $\Fcal_0$ are multilocal, so one can use the standard argument with the fundamental theorem of calculus to show that $\Fcal_{0}$ is equal to the image of $\delta_S$. I sketch it here for a local functional $F\in \Fcal_0$. I write $F$ as
\[
F(\ph)=\int_M \omega(u_1,\dots,u_k)d\mu=\int_M \tilde{\omega}(S_1,\dots,S_N,u_{N+1},\dots, u_{k})d\mu\,,
\]
where $u_1,\dots, u_{k}$ are some arbitrary fixed coordinates on the jet space. Since $F$ vanishes on the solution space, under the regularity assumption, we have $$\tilde{\omega}(0,\dots,0,u_{N+1},\dots, u_{k})=0\,.$$ 
We can then write
\[
F(\ph)=\sum_{\alpha=1}^{N}\int_M S_\alpha \int_{0}^{1} \frac{\partial\tilde{\omega}}{\partial S_\alpha}(\lambda S_1,\dots,\lambda S_N,u_{N+1},\dots, u_{k}) \, d\lambda d\mu\,.
\]
Next, one shows that the smoothness, locality and the fact that $\omega$ is compactly supported on $M$ imply that $ \int_{0}^{1} \frac{\partial\tilde{\omega}}{\partial S_\alpha(x)}(\lambda S_1,\dots,\lambda S_N,u_{N+1},\dots, u_{k}) \, d\lambda\equiv X^{\alpha}(x)$ defines a smooth local compactly supported vector field. I will always assume that the actions we consider satisfy the above regularity conditions.

Note that $\delta_S$ also ``knows'' about the symmetries, since the kernel of $\delta_S$ consists of those vector fields for which  
\[
\iota_{dS}(X)=\partial_X S\equiv 0\,,
\]
i.e. of \textit{symmetries}. We can summarize all we know up to now in the following chain complex:
	\[
	\begin{array}{c@{\hspace{0,2cm}}c@{\hspace{0,2cm}}c@{\hspace{0,2cm}}c@{\hspace{0,2cm}}c@{\hspace{0,2cm}}c@{\hspace{0,2cm}}c@{\hspace{0,2cm}}c@{\hspace{0,2cm}}c@{\hspace{0,2cm}}c}
	0&\xrightarrow{}&Sym&\hookrightarrow&\V&\xrightarrow{\delta_S}&\F&\rightarrow &0\\
	&&2&&1&&0&&&
	\end{array}
	\]
	where the numbers below mean \textit{grading} and they help to keep track of where things belong. 
	
	In homological algebra, \textit{homology groups} are defined by taking quotients of the kernel of the map going out of the space by the image of the map going into it. The 0th homology of our complex is $H_0=\F/\im(\delta_S)=\F/\F_0$, so it characterizes the space of  functionals on the solution space!
 Assume there are no non-trivial (not vanishing on $\Ecal_S$) local symmetries and let \textit{$\mathcal{K}\doteq\Big(\Lambda \V,\delta_S\Big)$} (this is the exterior algebra, built out of antisymmetrized tensor products). Then $\F_S=H_0(\mathcal{K})$ and higher homologies vanish.
 This is called the \textit{Koszul resolution}.

\subsubsection{Chevalley-Eilenberg complex}
Let's now consider a situation where local symmetries are present. Let $\frakg_c$ be the Lie algebra characterizing the infinitesimal local symmetries. Since we let them act as derivations on functionals that are themselves compactly supported, we can drop the requirement of compact support for the symmetries and consider $\frakg$ instead.

 We are interested in the space of symmetry-invariant observables, i.e. functionals $F$ such that
\[
\partial_{\rho(\xi)}F=0\,,
\]
for all $\xi\in\frakg$. Algebraically, the space of invariants under the action of a Lie algebra can be characterized using the 
\textit{Chevalley-Eilenberg complex}.

The underlying graded algebra of the Chevalley-Eilenberg complex is $\CEcal\doteq\Ci_\ml(\overline{\Ecal},\CC)$, the space of multilocal functionals on the graded manifold $\Ecal\oplus \frakg[1]\equiv \overline{\Ecal}$ (see \cite{Book} for a precise definition of this space). We identify functionals on $\frakg[1]$ with $\Lambda\frakg'$, the exterior algebra over $\frakg'$. The generators of $\Lambda\frakg'$ can be understood as evaluation functionals
\[
c^I(x)[\xi]\doteq \xi^I(x)
\]
and in physics these are called \textit{ghosts}. The grading of $\CEcal$ is called the \textit{ pure ghost number }$\#\pg$. We express $\CEcal$ as
$$\CE\!\doteq\big(\Lambda\frakg'\widehat{\otimes}\F,\gamma_{
\rm ce}\big)\,,$$
where $\widehat{\otimes}$ is the appropriately completed tensor product.

The Chevalley-Eilenberg differential $\gamma_{\mathrm{ce}}$ is constructed in such a way that it encodes the action $\rho$ of the gauge algebra $\frakg$ on $\Fcal$. For $F\in\Fcal$ we define $\gamma_{\mathrm{ce}} F\in \Ci_\ml(\Ecal,{\frakg'})$ as
\be\label{ChE1}
(\gamma_{\mathrm{ce}} F)(\ph,\xi)\doteq \partial_{\rho(\xi)}F(\ph)\,,
\ee
where $\xi\in\frakg$. In terms of evaluation functionals (i.e. ghosts):
\[
\gamma_{\mathrm{ce}} F=\partial_{\rho(c)} F\,.
\]
For a form $\omega\in \frakg'$, which doesn't depend on $\ph$ we set
$\gamma_{\mathrm{ce}} \omega(\xi_1,\xi_2)\doteq -\omega([\xi_1,\xi_2])$ which is an element of $\Lambda^2\frakg'$. Again, we can express this using evaluation functionals:
\[
\gamma_{\mathrm{ce}} c= -\frac{1}{2}[c,c]
\]
 If $F\in\F^\inv$  then $\gamma_{\mathrm{ce}} F\equiv 0$, so $H^0(\CE)$ characterizes the gauge invariant functionals.

\subsubsection{BV complex}
Now we combine gauge invariant and on-shell, to be able to characterize the space \textit{$\F_S^{\inv}$}. Note that $\CE$ is the space of multilocal compactly supported functionals on a graded manifold $\overline{\Ecal}=\Ecal\oplus \frakg[1]$, so instead of vector fields on $\Ecal$, we consider the vector fields on the \textit{extended configuration space}  $\overline{\Ecal}$. This way we obtain the \textit{BV complex}: $\BV$. Its underlying algebra is the algebra of \textit{multilocal polyvector fileds on $\overline{\Ecal}$}, i.e. the space of multilocal compactly supported functionals on the graded manifold
\[
\Ecal[0]\oplus\frakg[1]\oplus \Ecal^![-1]\oplus\frakg^![-2]\equiv  T^*[-1]\overline{\Ecal}\,.
\]
More concretely, elements of $\BVcal$ are multilocal functionals of the field multiplet $\ph^\alpha$ and of corresponding antifields $\ph_\alpha^\ddagger$, where the index $\al$ runs through all the physical and ghost indices. For graded functionals, we distinguish between the right $\frac{\delta_r }{\delta \ph^\al}$ and left derivatives $\frac{\delta_l }{\delta \ph^\al}$. We use the convention that the antifields are identified with the right derivatives.

The algebra $\BVcal$ has two gradings: the ghost number $\#\gh$ (the main grading) and the antifield number $\#\af$ (extra grading used later). Functionals of physical fields have both numbers equal to 0. Functionals of ghosts have $\#\af=0$ and $\#\gh=\#\pg$ (the ``pure ghost'' grading, a ghost $c$ has $\#pg=1$). All vector fields have a non-zero antifield number given by $\#\af(\ph_\alpha^\ddagger)=1+\#\pg(\ph^\alpha)$, and $\#\gh=-\#\af$.

$\BVcal$ seen as the space of graded multivector fields is equipped with a graded generalization of the Schouten bracket, called in this context \textit{the antibracket}, defined by
\be\label{eq:antibracket}
\{X,Y\}\doteq\sum_\alpha\left<\frac{\delta_r X}{\delta\ph^\alpha},\frac{\delta_l Y}{\delta\ph^\ddagger_\alpha}\right>-\left<\frac{\delta_r X}{\delta\ph^\ddagger_\alpha},\frac{\delta_l Y}{\delta\ph^\alpha}\right>\,.
\ee
The right derivation $\delta_S$ is not inner with respect to $\{.,.\}$, but locally it can be written as:
\[\delta_SX=\{X,L(f)\}\,, \quad f\equiv 1\ \textrm{on}\  \supp X\,,\ X\in \V\,.\]
We write this as $\delta_SX=\{X,S\}$. Similarly, one can find an action $\theta$ such that $\gamma X=\{X,\theta\}$ and we define the \textit{classical BV differential} as
\[
s=\{.,S+\theta\}\equiv\{.,S^{\ex}\}\,.
\]
We call $S^{\ex}$ the \textit{extended action}. The BV differential $s$ has to be nilpotent, i.e.: $s^2=0$, which leads to the \textit{classical master equation ({\cme})}:
\be\label{CME}
\{L^{\ex}(f),L^{\ex}(f)\}=0\,,
\ee
modulo terms that vanish in the limit of constant $f$.

The differential $s$ increases the ghost number by one (i.e. is of order 1 in $\#\gh$).
It can be expanded with respect to the antifield number as
\[
s=\delta+\gamma\,,
\]
where $\delta$ is of order -1 in $\#\af$ and is the extension of $\delta_S$, while $\gamma$ is of order 0 is the extension of $\gamma_{\mathrm{ce}}$. In general, there could be higher order terms as well, but I will not discuss this here.

Differential complex $(\BV,\delta)$ is called the \textit{Koszul-Tate complex} and in the simplest case discussed here, it is a resolution (it would not be a resolution if the symmetries were not independent).

Crucially, we have 
\be\label{eq:physical}
H^0(\BV,s)=\F_S^{\inv}
\ee
which is the reason to work with $\BV$ in the first place, as it contains the same information as $\F_S^{\inv}$, but has a simpler algebraic structure (quotients and spaces of orbits are resolved). 

To prove \eqref{eq:physical}, one uses the fact that the Koszul-Tate complex $(\BV,\delta)$ is a resolution (the only non-trivial homology is in degree $0$), so
\[
H^0(\BV, s)=H^0(H_0(\BV,\delta), \gamma)\,.
\]
Since $H_0(\BV,\delta)$ is by construction the space of on-shell functionals on $\overline{\Ecal}$ and the 0-th cohomology of $\gamma$ characterizes the invariants, we obtain the desired result.

In the next step, we introduce the gauge fixing using an automorphism $\alpha_\Psi$, defined on generators as
\[
\alpha_\Psi(\Phi^\ddagger_\beta(x))\doteq \frac{\delta \Psi(f)}{\delta \ph^\beta(x)}\,,\quad \alpha_\Psi(\Phi^I(x))=\Phi^I(x)
\] 
where $f(x)=1$ and $\Psi_M(f)$ is a fixed generalized Lagrangian of ghost number -1, called \textit{ gauge fixing fermion}. The choice of $\Psi_M$ determines the choice of gauge fixing. It can be easily seen that $\alpha_\Psi$ leaves the antibracket invariant and we choose it in such a way that the $\#\af=0$ part of the transformed action gives rise to hyperbolic equations (see \cite{FR} for details). 
\begin{exa}
	To implement a Lorenz-like gauge in Yang-Mills theory, we need to further extend the BV complex with antighosts $\bar{C}$ (in degree -1) and Nakanishi-Lautrup fields $B$ (in degree 0). These form a trivial pair, i.e.:
	\[
	s\bar{C}^I=iB^I\,\quad s B^I=0
	\]
	The new extended configuration space is written explicitly as
	\[
	\overline{\Ecal}=\Ecal\oplus\frakg[1]\oplus\frakg[0]\oplus\frakg[-1]\,.
	\]
	Since the new generators were introduced as a trivial pair, the cohomology of the resulting complex is the same as the original one, so \eqref{eq:physical} remains true also after this modification. The gauge-fixing fermion is then:
	\[
	\Psi_M(f)=i\int\limits_M f\left(\frac{\alpha}{2}\kappa(\bar{C},B)+\left<\bar{C},*d*\!A\right>\right)d\mu
	\]	
\end{exa}
To talk about the gauge-fixed theory, it is convenient to redefine the gradings. Let $\#\ta$ denote the \textit{total antifield number}, which is 1 for all the antifield generators and zero for fields. We decompose $s$ with respect to this grading and obtain two terms (which I again denote by $\delta$ and $\gamma$)
\[
s=\gamma+\delta\,,
\]
The total action is still denoted by $S^{\ex}$ and I will denote by $S$
the $\#\ta=0$ term in the action. Let $\theta:=S^{\ex}-S$. We can express
\[
\delta=\{.,S\}\,,\quad \gamma=\{.,\theta\}\,.
\]
Differential $\delta$ acts trivially on fields and on antifields it gives $\delta \ph^\ddagger_{\alpha}=\frac{\delta S}{\delta \ph_\alpha}$, so the \textit{gauge-fixed} equations of motion are now the equations of motion of $S$, which are hyperbolic. This implies that the homology of $\delta$ is concentrated in degree 0 (there can be no non-trivial local symmetries for hyperbolic equations!), so $(\BV,\delta)$ is a resolution and we again have
\[
\Fcal^{\inv}_S=H^0(\BV,s)=H^0(H_0(\BV,\delta),\gamma)\,.
\]
Taking $H_0(\BV,\delta)$ is interpreted as ``going on-shell''.
\subsection{Linearized theory}
We can split the extended action into the term $S_0$ that is quadratic in fields and antifields and the interaction therm $V$. $S_0$ can be written as 
\[
S_0=S_{00}+\theta_0\,,
\]
where $S_{00}$ is the term with $\#\ta=0$ and $\theta_0$ has  $\#\ta=1$. Similarly $V=V_{0}+\theta$ and we note that $S=S_{00}+V_{0}$ is the total antifield independent part of the action.

We define the linearized BRST differential by
\[
\gamma_0 F\doteq \{F,\theta_0\}\,,
\]
The total linearized BV differential $s_0$ is
\[
s_0=\delta_0+\gamma_0\,,
\]
where $\delta_0(\ph^\ddagger_\alpha)=-\frac{\delta S_{00}}{\delta\ph^\alpha}$, so the homology of $\delta_0$ describes the space of solutions to the linearized equations of motion. Denote 
$$\frac{\delta_l S_{00}}{\delta\ph^\al(x)}(\ph)\equiv P_{\alpha\beta}(x)(\ph^\beta(x))\,,$$
where each component $P_{\al\beta}$ is a differential operator. For simplicity, we will often write the equations of motion using the index-free notation: $P\ph=0$.

In a similar manner, I will denote
 $$\frac{\delta_r\delta_l \theta_0}{\delta\ph^\sigma(y)\delta\ph_\al^\dgr(x)}\equiv K^{\al}_{\phantom{\al}\sigma}(x)\delta(y-x)\,,$$
  where each $K^{\al}_{\ \sigma}$ is a differential operator.

The cohomology of $s_0$ is given by
\[
H^0(\BV,s_0)=H^0(H_0(\BV,\delta_0),\gamma_0)\,,
\]
since $(\BV,\delta_0)$ is a resolution. Taking $H_0(\BV,\delta_0)$ is again understood as ``going on shell''  (this time for the linearized theory).\footnote{Here we note a difference with \cite{H}, where, in the example of Yang-Mills theory, the terms $\nabla_\mu {A_I^\ddagger}^\mu$ in $s_0(c_I^\ddagger)$ and $\overline{C}^\ddagger_I$ in $s_0(B^\ddagger)$  were attributed to the ``$\delta_0$ part'' of $s_0$ rather than to the ``$\gamma_0$ part''. We will denote the operators used in \cite{H} by
$\tilde{\gamma}_0$ and $\tilde{\delta}_0$, with $s_0=\tilde{\gamma}_0+\tilde{\delta}_0$.
By direct inspection one can see that $\tilde\delta_{0}$ does not respect the total antifield grading (since for example $\tilde\delta_{0}(C_I^\ddagger)=id*d\overline{C}_I-dA^\ddagger_I$ has a term with $\#\ta=0$ as well as a term with $\#\ta=1$). Moreover, $\tilde\delta_{0}$ is not nilpotent
and it does not anti-commute with $\tilde\gamma_0$.
Hence, from the cohomological perspective, using $\delta_0$ and $\gamma_0$ is more natural than using $\tilde{\delta}_0$ and $\tilde{\gamma}_0$.}

Assume that the gauge fixing was done in such a way that $P$ is Green hyperbolic (for gauge theories and gravity this was shown in \cite{FR}), meaning that there exist unique retarded and advanced Green functions $\Delta^{\rm A/R}$, i.e. Green functions for the equations of motion operator $P$ such that
\[
\supp(\Delta^{\rm R}(f))\subset  J^+(\supp(f))\,,\qquad
\supp(\Delta^{\rm A}(f))\subset J^-(\supp(f))\,.
\]
We define the Pauli-Jordan function by
\[
\Delta=\Delta^{\rm R}-\Delta^{\rm A}\,.
\]
%
The {\cme} of the free theory allows one to prove some important properties that hold for $\Delta^{\rm A/R}$ and $\Delta$ (see e.g. \cite{H,Rej13}).
\begin{lemma}\label{gauge:inv:Delta}
	Assume that $S_{00}$ is invariant under $\gamma_0$ (i.e. the free {\cme} holds) and $S_{00}$ induces a normally hyperbolic system of equations of motion: $P\ph=0$. Let $\Delta^*$ be a retarded, advanced or causal propagator corresponding to $P$. It follows that $\Delta^*$ satisfies the ``consistency conditions'' (see \cite{H}):
	\be\label{const:cond}
	\sum_\sigma((-1)^{|\ph^\al|}K^{\al}_{\ \sigma}(x')\Delta^*(x',x)^{\sigma\gamma}+K^{\gamma}_{\ \sigma}(x)\Delta^*(x',x)^{\al\sigma})=0\,,
	\ee
\end{lemma}
\begin{exa}
For Yang-Mills, these identities are:
\[
K^{A}_C(x)\Delta^{*}(x,y)^{C\overline{C}}+K^{\overline{C}}_B(y) \Delta^{*}(x,y)^{A,B}=0\,,\qquad
-K^{\overline{C}}_B(x)\Delta^*(x,y)^{BA}+K^A_C(y)\Delta^*(x,y)^{\overline{C} C}=0\,,
\]
or more explicitly:
\begin{align*}
d_x \Delta_s^*(x,y)+\delta_y\Delta_v^*(x,y)&=0\,,\\
\delta_x \Delta_v(x,y) +d_y\Delta_s^* (x,y)&=0\,.
\end{align*}
These relations are an obvious consequence of the fact that $d$ and $\delta$ commute with the Hodge Laplacian.
\end{exa}

The classical linearized theory is constructed by introducing the Peierls bracket given by:
\begin{equation}\label{eq:Peierls}
\Pei{F}{G} = \sum_{\al,\beta} \skal{\frac{\delta^r F}{\delta\ph^\al}}{{\De}^{\al\beta}\frac{\delta^l G}{\delta\ph^\beta}},
\end{equation}
where $F, G \in\BV$. Unfortunately, $\BV$ is not closed under this bracket and one needs to extend it to a larger space. A good candidate is the space $\BV_\mc$ of \textit{microcausal} functions on $T^*[-1]\overline{\Ecal}$, i.e. functionals that are smooth, compactly supported and their derivatives (with respect to both $\ph$ and $\ph^\dgr$) satisfy the WF set condition:
\be\label{mlsc}
\WF(F^{(n)}(\ph,\ph^\dgr))\subset \Xi_n,\quad\forall n\in\NN,\ \forall\ph\in\overline{\E}(M)\,,
\ee
where $\Xi_n$ is an open cone defined as 
\be\label{cone}
\Xi_n\doteq T^*M^n\setminus\{(x_1,\dots,x_n;k_1,\dots,k_n)| (k_1,\dots,k_n)\in (\overline{V}_+^n \cup \overline{V}_-^n)_{(x_1,\dots,x_n)}\}\,,
\ee

\subsection{Classical BV operator and the M{\o}ller maps}
The key observation of \cite{FR3} is that one can \textit{define} the interacting quantum BV operator by taking the free one and twisting it with the quantum M{\o}ller map. Then, one has to prove that the resulting map is \textit{local}. The advantage of this viewpoint is that one separates the question of definition and existence of the quantum BV operator from the particular technical results one needs to establish its locality. The latter is crucial from the physical viewpoint, but mathematically, one could very well just go ahead with the non-local operator. To understand this idea better, it's good to first have a look at the classical case.

In the classical limit, the interacting classical BV operator should arise from a twist of the free one with the classical M{\o}ller map. Here we provide the direct proof that this is indeed the case.

Consider the theory with action $S=S_0+V$. For simplicity, we will treat  the interaction as a local compactly supported functional rather than a generlized Lagrangian. We will also omit the test function in $S_0$. 

For $S$ without non-trivial local symmetries, the inverse M{\o}ller map is defined as \cite{DF02} (see also \cite{HR}):
\be\label{class:M}
r_{\lambda V}^{-1}(F)(\ph)\doteq F(
{\mathtt r}^{-1}_{\la V}(\ph))\,,
\ee
where 
\be\label{class:M:field}
{\mathtt r}^{-1}_{\la V}(\ph) = \ph + \la \Delta^{\mathrm{R}} V^{(1)}(\ph) \,.
\ee
and  $\Delta^{\rm R}$ is the retarded Green function of the free theory. It can then be inverted as a formal power series to obtain the classical M{\o}ller map $r_{\lambda V}$, whereupon
\be\label{YangFeldmann}
{\mathtt r}_{\la V}(\ph)=\ph-\la\Delta_{S_0}^{\mathrm{R}}V^{(1)}({\mathtt r}_{\la V}(\ph))\,,
\ee
which is the Yang-Feldmann equation.

In this way of defining things, $r_{\lambda V}$  goes from the interacting to free theory and the image of $r_{\lambda V}$ represents the interacting fields constructed from the free ones. We also have the intertwining relation:
\[
\Pei{r_{\lambda V}F}{r_{\lambda V}G}=r_{\lambda V}\Pei{F}{G}_{V}\,,
\]
where  $\Pei{.}{.}$ and  $\Pei{.}{.}_{V}$ are the free and the interacting Poisson bracket, respectively. Moreover, it is easily seen that $r_{\lambda V}^{-1}$ maps the ideal generated by the free equations of motion to the ideal generated by the interacting equations of motion, i.e.:
\[
r_{\lambda V}^{-1}\frac{\delta S_0}{\delta \ph }=r_{\lambda V}^{-1}(P\ph)=P\ph+\lambda P\circ \Delta^{\rm R} V^{(1)}(\ph)=P\ph +\lambda V^{(1)}(\ph)\,.
\]
It is, therefore compatible with taking the quotients by the both ideals.

In the BV-extended version, we set ${\mathtt r}_{\la V}$ to act trivial on antifields and the result above about intertwining the ideals implies that
\be\label{Moller1}
r^{-1}_{\lambda V}\circ \delta_0=\delta \circ r^{-1}_{\lambda V}\,.
\ee
\begin{rem}
Here some caution is required. Since $r_{\lambda V}$ is a non-local map, the statement above \textit{does not imply} that the local cohomologies of $\delta$ and $\delta_0$ are the same! When restricted to local functionals, $\delta_0$ and $\delta$ yield different cohomologies, as one would expect. In the literature, one always computes the local cohomologies of $\delta_0$ and $\delta$, so the relation \eqref{Moller1} between these two operators has been apparently overlooked an might seem rather surprising on first sight.
\end{rem}

We now move on to the more complicated case, where gauge symmetries are present. First of all, now $S_0$ has two terms, one of which, $S_{00}$ does not depend on the antifields and this is the term that defines $P$ (and hence $\Delta^{\rm R}$). 

The formula for the M{\o}ller operator is the same as in the scalar case \eqref{class:M}, but we need to replace \eqref{class:M:field} with
\[
{\mathtt r}^{-1}_{V}(\ph^\alpha) = \ph^\alpha +  (\Delta^{\mathrm{R}})^{\alpha\beta}\, \frac{\delta_l V}{\delta\ph^\beta}(\ph) \,,
\]
and \eqref{YangFeldmann} with 
\[
{\mathtt r}_{V}(\ph^\alpha) = \ph^\alpha - (\Delta^{\mathrm{R}})^{\alpha\beta}\, \frac{\delta_l V}{\delta\ph^\beta}({\mathtt r}_{V}(\ph)) \,.
\]
\begin{thm}\label{MollerBV}
	Let $X\in\BV$ and assume that $S_0$ satisfies the classical master equation, then 
	\[
	r_{V}^{-1}(\{X,S_0\})  =\{r_{V}^{-1}(X),S_0+V\}-\int \frac{\delta_r X}{\delta\ph^\gamma(y)}({\mathtt r}^{-1}_{ V}(\ph))\Delta^{\rm R}(y,z)^{\gamma\beta}\frac{\delta_l}{\delta\ph^\beta(z)}(\cme(S))\,,
	\]
	where $\cme(S)$ is the classical master equations for the full theory.
\end{thm}
\begin{cor}
	From theorem \ref{MollerBV} follows that, assuming {\cme}, we can write the classical BV operator of the full theory as
	\[
	s=r_V^{-1}\circ s_0\circ r_V\,,
	\]
	which is the classical analog of the definition of the quantum BV operator proposed in \cite{FR3}.
\end{cor}
The proof of the theorem is rather technical, so we present it in the appendix.
\section{Quantization}
\subsection{Free theory}\label{free:theory} 
The quantized algebra of free fields is constructed by means of \textit{deformation quantization} of the classical algebra $(\BV_{\mc},\Pei{.}{.})$. To this end, we equip the space of formal power series $\BV_{\mc}[[\hbar]]$ with a noncommutative star product which corresponds to the operator product of quantum observables. 

Define  the $\star$-product (deformation of the pointwise product) by
\begin{equation*}
F\star G\doteq m\circ \exp({i\hbar D_W})(F\otimes G),
\end{equation*}
where $m$ is the multiplication operator, i.e. $m(F\otimes G)(\ph)=F(\ph)G(\ph)$, and  $D_W$ is the functional differential operator defined by
\begin{equation*}
D_W\doteq \frac{1}{2} \sum_{\al, \beta} \left<{W}^{\al\beta},\frac{\delta^l}{\delta\ph^\al} \otimes \frac{\delta^r}{\delta\ph^\beta}\right>\,.
\end{equation*}
with $W$, the \textit{2-point function of a Hadamard state}. $W$ is positive definite, satisfies the appropriate wavefront set condition \cite{Rad} and we have \textit{$W=\frac{i}{2}\Delta+H$}, where $H$ is a symmetric bisolution for $P$. In addition to these standard properties, we also need to require the\textit{ consistency condition }\cite{H} on the symmetric part:
\be\label{const:cond2}
\sum_\sigma((-1)^{|\ph^\al|}K^{\al}_{\ \sigma}(x')H(x',x)^{\sigma\gamma}+K^{\gamma}_{\ \sigma}(x)H(x',x)^{\al\sigma})=0\,,
\ee
Note that this is automatically fulfilled for $\Delta$ (see  \eqref{const:cond} ). Under this condition, $\gamma_0$ is a right derivation with respect to the star product. Since $W$ is a solution for the linearized equations of motion operator $P$, $\delta_0$ is also a right derivation with respect to $\star$. We can therefore conclude that
\[
s_0(X\star Y)=(-1)^{\#\gh (Y)}s_0X\star Y+X\star s_0Y\,.
\]
\subsection{Interacting theory}
\subsubsection{Time-ordered products}
For simplicity, we start our discussion by considering \textit{regular functionals} (functionals whose derivatives at every point are smooth compactly supported functions) on $T^*[-1]\overline{\Ecal}$. We use the notation $\BV_{\reg}$.

The \textit{time-ordering operator $\TT$} is defined as:
	\[
	\TT F(\ph)\doteq e^{\frac{\hbar}{2}\Dcal_{\Delta^{\rm F}}}\ ,
	\]
where, for an integral kernel $K$, we define
$$\Dcal_K\doteq \sum_{\al,\beta}\left<{K}^{\al\beta}, \frac{\delta^l}{\delta\ph^\al}\frac{\delta^r}{\ph^\beta}\right>$$
and
$\Delta^{\rm F}=\frac{i}{2}(\Delta^{\rm A}+\Delta^{\rm R})+H$. 

Formally, $\TT$ corresponds to  the operator of convolution with the oscillating Gaussian measure ``with covariance $i\hbar\Delta^{\rm F}$'',
	\[
	\TT F(\ph)\form \int F(\ph-\phi)\, d\mu_{i\hbar\Delta_F}(\phi) \ . 
	\]
We define the \textit{time-ordered product} $\T$ on $\BV_\reg[[\hbar]]$ by:
\[
		F\T G\doteq \Tcal(\Tcal^{\minus}F\cdot\Tcal^{\minus}G)
		\]
		\begin{rem}
Note that $\T$ is the time-ordered version of $\star$, in the sense that
$F\T G=F\star G$ if the support of $F$ is later than the support of $G$ and $F\T G=G\star F$, if the support of $G$ is later than the support of $F$.
\end{rem}
\subsubsection{Peierls bracket from the antibracket}
Before I continue with the interacting theory, I would like to address one more issue, often omitted in the literature: the precise relation between the antibracket \eqref{eq:antibracket} and the Peierls bracket \eqref{eq:Peierls}. The example of the scalar field has been discussed in \cite{GR20}. Here I give the general statement.

First of all, I need to introduce one more key concept from the AQFT axiomatic framework, which is yet another way to describe the dynamics of the theory. In AQFT, a QFT model is specified by assigning algebras of observables $\fA(\Ocal)$ to relatively compact regions $\Ocal\subset \Mcal$ of a given spacetime. In the original work of Haag and Kastler \cite{HK} these algebras were assumed to be $C^*$-algebra, but to allow for the use of perturbative methods and homological algebra, one has to weaken this assumption. In \cite{GR20}, together with Gwilliam, we use instead chain complexes in associative, unital $*$-algebras. Being unital means that they have a unit (one can think of it as the identity operator) and $*$ is an involution, the abstract notion of taking the adjoint of an operator. For a more general formulation using homotopical algebra see \cite{BSW19}.

\begin{exa}
	As an example, consider $\fA(\Ocal)=(\BVcal(\Ocal)[[\hbar]],\star)$, as defined in the previous section, where $\BVcal(\Ocal)$ is obtained by restricting to functionals supported inside a relatively compact $\Ocal\subset M$. It is clearly an associative algebra (the product is $\star$), the unit is the constant functional $1$ and the involution is complex conjugation. It is also a chain complex, where the differential is $s_0$. 
\end{exa}

Let $\Ocal\mapsto \fA(\Ocal)$, with differential $d$ and product $\star$, be an assignment of such chain complexes in algebras to regions. There are two important axioms to impose here:
\begin{itemize}
	\item {\bf Einstein causality:} for $\Ocal_1,\Ocal_2$ that are spacelike to each other, the commutar
	$[\fA(\Ocal_1),\fA(\Ocal_2)]=dX$ for some $X\in\fA(\Ocal)$ for any $\Ocal$ that contains both $\Ocal$.
	\item {\bf Time-slice axiom:} for any $\Ncal$ a neighborhood of a Cauchy surface\footnote{A hypersurface in $M$ such that every inextendible causal curve intersects it exactly once. Cauchy surfaces are used for formulating initial-value problems for normally hyperbolic operators, e.g. the wave operator.} in the region $\Ocal\subset M$, 
	the map $\fA(\Ncal)$ and  $\fA(\Ocal)$ are quasi-isomorphic, i.e. isomorphic on the level of cohomology groups.
\end{itemize}
The first axiom is the weak version of causality. The second is the quantum analog of well-posedness of the Cauchy problem.

Recall that $\star$ arises from the deformation of the Peierls bracket $\Poi{.}{.}$ \eqref{eq:Peierls}, so we will start in the quantum theory. Assume we have a theory $\fA$ with the product $\star$ and the differential $d$ that obeys the time-slice axiom and we have a time-ordered product $\T$ associated with $\star$. Take $F,G\in \fA(\Ocal)$ and consider Cauchy surfaces to the future and to the past of $\Ocal$, denoted $\Ncal_+$ and $\Ncal_-$. The time-slice axiom implies that there exist maps $\beta_-$ and $\beta_+$ such that
$\beta_+(F)$ is localized in the future of $\Ocal$ and $\beta_-(F)$ in the past. Using the time-slice axiom, modulo the image od $d$, we can write the $\star$-commutator of $F$ and $G$ as
\[
i\hbar [G,F]_{\star}=G\star F-F\star G=G\star \beta_+(F)-\beta_-(F)\star G
=G\T \beta_+(F)-\beta_-(F)\T G\quad \mathrm{mod}\ \mathrm{Im}d
\]
From the time-slice axiom follows also that there exists $\Psi$ such that $\beta_-F-\beta_+F=s_0\Psi$. Hence we rewrite the $\star$ commutator as
\[
	[G,F]_{\star}=G\star \beta_+(F)-\beta_-(F)\star G
	=G\T \beta_+(F)-\beta_-(F)\T G\\=G\T(\beta_-F-\beta_+F)=G\T d\Psi\,, 
\]
for some $\Psi$. Therefore, we can express the Peierls bracket as
\[
i\hbar\Poi{G}{F}=G\T s_0\Psi\quad \mathrm{mod}\ \hbar^2\,,\textrm{Im} d\,.
\]
Assume that $d G=0$. We can re-write the right-hand side using the antibracket as follows:
\[
i\hbar\Poi{G}{F}=s_0(G\T \Psi) +i\hbar \{G,\Psi\} \quad \mathrm{mod}\ \hbar^2\,,\textrm{Im} d\,.
\]
Hence
\[
\Poi{G}{F}=\{G,\Psi\} \quad \mathrm{mod}\ \hbar\,,\textrm{Im} \,,,
\]
which can be thought of as the \textit{intrinsic definition of the Peierls bracket, given the antibracket and the time-ordered product in a theory satisfying time-slice axiom}.
\subsubsection{Interaction}
We model interactions as functionals $V$ and for the moment assume $V\in\BV_\reg$. We define the quantum observable (of the free theory), associated with $V$, as $\TT V$. In the language of deformation quantization, we can say that we use $\TT$ as the \textit{quantization map}.  By analogy to normal ordering, we use the notation $\TT V\equiv \no{V}$.

We define the \textit{formal S-matrix}, $\Scal(\lambda \no{V})\in\BV_{\reg}((\hbar))[[\lambda]]$ by
	\[
	\Scal(\lambda V)\doteq e_{\sst{\TT}}^{i\lambda\no{V}/\hbar}=\TT(e^{i\lambda V/\hbar})\,.
	\]
 \textit{Interacting fields} are elements of $\BV_{\reg}[[\hbar,\lambda]]$ given by
	\[
	R_{\la V}(F)\!\doteq\! (e_{\sst{\TT}}^{i\la \no{V}/\hbar})^{\star \minus}\star (e_{\sst{\TT}}^{i\la \no{V}/\hbar}\T \no{F})=-i\hbar\frac{d}{d\mu}\Scal(\la V)^{-1}\Scal(\la V+\mu F) \big|_{\mu=0}
	\]
For $\lambda=0$, we recover $R_{0}(F)=\no{F}$. We define the \textit{interacting star product} as:
	\[
	F\star_{int} G\doteq R_V^{\minus}\left(R_V(F)\star R_V(G)\right)\,,
	\]
\subsubsection{Renormalization problem}
The problem that one faces is that interesting interactions and observables are local, but not regular. In fact, polynomial local functionals of order greater than one cannot be regular, as illustrated in the example below.
\begin{exa}
	Consider the free scalar field and the functional
	$$F(\ph)=\int f(x,y)\ph(x)\ph(y)d\mu(x) d\mu(y)\,,\ \  f\in\Dcal(M^2):= \Ci_c(M^2,\RR)\,,$$
	which is regular. Now constrast it with
	$$F(\ph)=\int f\ph^2d\mu=\int f(x)\delta(x-y)\ph(x)\ph(y)d\mu(x) d\mu(y)\,,$$
	which is local, but fails to be regular, since the second derivative is:
	\[
	F^{(2)}(\ph)(x,y)=f(x)\delta(x-y)d\mu(x) d\mu(y)\,,
	\]
	i.e. it is not smooth.
\end{exa}
Because of singularities of $\De^{\rm F}$, the time-ordered product $\T$ is not well defined on local, non-linear functionals, but the \textit{physical interactions are usually local}!

The  \textit{renormalization problem} is then to extend $\Scal$ to local arguments by extending time-ordered products:
\[
\Scal(V)=\sum\limits_{n=0}^\infty \frac{1}{n!} \TT_{n}(V,...,V)\,.
\]
We note that the time-ordered product $\TT_{n}(F_1,...,F_n)\doteq F_1\T...\T F_n$ of $n$ local functionals is well defined if their supports are pairwise disjoint. To extend  $\TT_{n}$ to arbitrary local functionals we use the causal approach of Epstein and Glaser (causal perturbation theory). The crucial property one uses in this process is the  \textit{causal factorization property}: if the supports of $F_1\ldots F_k$ are later than the supports of $F_{k+1},\ldots F_n$, then
\be\label{CFP}
\TT_{n}(F_1\otimes \dots \otimes F_n)=
\TT_{k}(F_1\otimes \dots \otimes F_k) \star
\TT_{n-k}(F_{k+1} \otimes \dots \otimes F_n) \, ,
\ee

\subsection{{\qme} and the quantum BV operator}
In the framework of \cite{FR3}, an important role is played by the condition that the S-matrix is invariant under the free classical BV operator:
\be\label{invmatrix}
s_0\left(e_{\sst{\TT}}^{i \no{V}/\hbar}\right)=0\,,
\ee
There is a very useful identity satisfied by $\TT$, namely:
\be\label{qbv1}
\delta_{0}(\TT F)=\TT(\delta_0 F-i\hbar\Lap F)\,,
\ee
where $\Lap$ is the \textit{BV Laplacian}, defined by:
	\be\label{BVLap}
\Lap X=(-1)^{(1+\#\gh(X))}\sum_\alpha\int dx \frac{\delta_r^2 X}{\delta\ph^\alpha(x)\delta\ph^\ddagger_\alpha(x)}\,.
\ee
Moreover, from the consistency conditions \eqref{const:cond2} follows that
\be\label{qbv2}
\TT\circ \gamma_0=\gamma_0\circ \TT
\ee
Putting these two together, we note that the left-hand side of \eqref{invmatrix} can be rewritten as:
	\[
s_0\left(e_{\sst{\TT}}^{i \no{V}/\hbar}\right)=\TT \left(s_0 e^{i  V/\hbar}-i\hbar \Lap e^{i  V/\hbar} \right)=\TT \left( e^{iV/\hbar} \left(\frac{i}{\hbar } \{V,S_0\}+\frac{i}{2\hbar}\{ V,V\}  +\Lap(V)\right)
\right)	\]
Setting $\Lap S_0=0$ (for symmetry reasons) and using the classical master equation, we can conclude that
\[
s_0\left(e_{\sst{\TT}}^{i \no{V}/\hbar}\right)=\frac{i}{\hbar}\,e_{\sst{\TT}}^{i \no{V}/\hbar}\T \TT\left(\frac{1}{2}\{S_0+V,S_0+V\}-i\hbar\Lap(S_0+V)\right)
\]
and we observe that the condition \eqref{invmatrix} is in fact equivalent to the \textit{quantum master equation} ({\qme}):
	\[
		\frac{1}{2}\{S_0+V,S_0+V\}=i\hbar\Lap(S_0+V)\,,
		\]
understood as a condition on $V$, which turns out to be important for the locality of the \textit{quantum BV operator}. In the free theory, we define it as follows:
\be\label{quantumBV}
\hat{s}_0\doteq \TT^{-1}\circ s_0 \circ \TT\,,
\ee
so from \eqref{qbv1} and \eqref{qbv2} follows that
\[
\hat{s}_0=s_0-i\hbar \Lap\,.
\]
In the interacting theory, the quantum BV operator $\hat{s}$ is defined on regular functionals by:
		\[
		\hat{s}=R_V^{-1} \circ s_0\circ R_V\,,
		\]
so it is the twist of the free classical BV operator by the \textit{(non-local!)} map that intertwines the free and the interacting theory. The classical limit of this definition makes sense, as demonstrated in theorem \ref{MollerBV}.

The 0th cohomology of $\hat{s}$ characterizes \textit{quantum gauge invariant observables}. Assuming \qme, 
\[
\hat{s}F=e_{\sst{\TT}}^{-i\no{V}/\hbar}\T s_0\left(e_{\sst{\TT}}^{i\no{V}/\hbar}\T\no{F}\right)= \{F,S_0+V\}-i\hbar \Lap(F)=s_0-i\hbar \Lap(F) \,.
\]
The second equality is particularly striking, since it shows that $\hat{s}$ is \textit{local}. In contrast to other frameworks, in our approach {\qme}
is not necessary for the nilpotency of $\hat{s}$ (this is true by definition), but is crucial for its \textit{locality}. 
\subsection{Renormalized {\qme} and quantum BV operator}
	To extend {\qme} and $\hat{s}$ to local observables, I replace now $\T$ with the renormalized time-ordered product.
	\begin{thm}[\cite{FR3}]
		The renormalized time-ordered product $\TR$ is an associative product on $\TTR(\F)$ given by
		\[
		F\TR G\doteq\TTR(\TTR^{\minus}F\cdot\TTR^{\minus}G)\,,
		\]
		where $\TTR:\F[[\hbar]]\rightarrow\TTR(\F)[[\hbar]]$ is defined as
		\[
		\TTR=(\oplus_{n}\TTR^{\,n})\circ\beta\,,
		\]
		where $\beta:\TTR:\F\rightarrow S^\bullet\F^{(0)}_\loc$ is the inverse of multiplication $m$ and we set $\TTR\big|_{\Fcal_\loc}=\id$ (so $\no{V}=V$).
	\end{thm}
Since $\TR$ is an associative, commutative product, we can use it in place of $\T$ and define the renormalized {\qme} and the quantum BV operator using formulas \eqref{invmatrix} and \eqref{quantumBV}.
These formulas get even simpler if we use the \textit{anomalous Master Ward Identity} \cite{BreDue,H}:
\be\label{MWI:inf:antif}
s_0(e_\TTR^{i\no{V}/\hbar})\equiv \{e_{\sst{\TTR}}^{iV/\hbar},S_0\}=\frac{i}{\hbar}e_{\sst{\TTR}}^{iV/\hbar}\TR(\tfrac{1}{2}\{V+S_0,V+S_0\}_{\TTR}-i\hbar \Lap_V)\,,
\ee
where $\Lap_{V}$ is identified with the \textit{anomaly term}.
If $S_0$ does not depend on antifields, \eqref{MWI:inf:antif} reduces to:
\be\label{MWI:inf}
\int \left(e_{\sst{\TTR}}^{iV/\hbar}\TR \frac{\delta V}{\delta \ph^\ddagger(x)}\right)\star\frac{\delta S_0}{\delta\ph(x)}=e_{\sst{\TTR}}^{iV/\hbar}\TR(\tfrac{1}{2}\{V+S_0,V+S_0\}_{\TTR}-i\hbar\Lap_V)\,,
\ee
\begin{rem}
Note that for regular $V$ one has
\begin{align*}
s_0(e_\TT^{i\no{V}/\hbar})&=\TT(\hat{s}_0e^{iV/\hbar})=\TT(s_0e^{iV/\hbar}-i\hbar \Lap e^{iV/\hbar})\\
&=\frac{i}{\hbar}e^{i\no{V}/\hbar}\T\TT\left(\{V,S_0\}+\frac{1}{2}\{V,V\}-i\hbar\Lap V\right)
=\frac{i}{\hbar}e^{i\no{V}/\hbar}\T\TT\left(\frac{1}{2}\{S_0+V,S_0+V\}-i\hbar\Lap V\right)\,,
\end{align*}
so the MWI is the renormalized version of this identity.
\end{rem}
The \textit{renormalized quantum master equation} is therefore:
\[
\tfrac{1}{2}\{V+S_0,V+S_0\}-i\hbar \Lap_V=0\,.
\]
Replacing $V$ with $V+\lambda F$ in \eqref{MWI:inf:antif} and differentiating with respect to $\lambda$ leads to the following identity for the classical BV operator:
\[
s_0(e_{\TTR}^{i\no{V}/\hbar}\TR F)=\frac{i}{\hbar}e_{\sst{\TTR}}^{iV/\hbar}\T\TTR\left(\{F,V+S_0\}_{\TTR}-i\hbar\Lap_VF+\frac{i}{\hbar} F \left(\frac{1}{2}\{S_0+V,S_0+V\}-i\hbar\Lap_V\right)\right)
\]
where $\Lap_V(F)\doteq \tfrac{d}{d\lambda} \Lap_{V+\lambda F}\big|_{\lambda=0}$. Assuming that the renormalized {\qme} holds, this reduces to:
\[
s_0(e_\TTR^{i\no{V}/\hbar}\TR F)=\frac{i}{\hbar}e_{\sst{\TTR}}^{iV/\hbar}\TR\TTR\left(\{F,V+S_0\}_{\TTR}-i\hbar\Lap_VF\right)\,,
\]
so the renormalized BV operator takes the form:
\[
			\hat{s}F=\{F,V+S_0\}-i\hbar\Lap_V(F)\,,
\]
 Hence, by using the renormalized time ordered product $\TR$, we obtained in place of $\Lap(X)$, the interaction-dependent operator $\Lap_V(X)$ (the anomaly). It is of order $\mathcal{O}(\hbar)$ and local.
 In the renormalized theory, $\Lap_V$ is well-defined on local vector fields, in contrast to $\Lap$.
\section{Towards a non-perturbative formulation}
\subsection{Local S-matrices}
In a recent paper \cite{BF19} Buchholz and Fredenhagen have shown that one can formulate interacting quantum theory of the scalar field in terms of local S-matrices $\Scal$, treated as a family of unitaries labelled by local functionals, generating a $C^*$-algebra. One then imposes relations that the S-matrices satisfy.

Let $F_1,F_2$ be local functionals and let $F_1\prec F_2$ denote the relation: $\supp F_1$ is not to the future of $\supp F_2$ (i.e. $\supp F_1$ does not intersect $J^+(\supp F_2)$). Local S-matrices are required to satisfy:
	\begin{enumerate}[{\bf S1}]
		\item {\bf Identity preserving}:\label{S:start} $\Scal(0)=\1$.
		\item {\bf Locality}:\label{S:loc} $\Scal$ satisfies the Hammerstein property, i.e. 
		$F_1\prec F_2$ implies that
		\[
		\Scal(F_1+F+F_2)=	\Scal(F_1+F)\Scal(F)^{-1}	\Scal(F+F_2)\,,
		\]
		where $F_1, F, F_2\in \F_\loc$.
	\end{enumerate}


Using time-ordered products and star products, one can construct a concrete realization of local S-matrices.  Following \cite{BF19}, I will denote by $\fA$ the group algebra over $\CC$ of the free group generated by elements $\Scal(F)$, $F\in \Fcal_{\loc}$, modulo relations {\bf S\ref{S:loc}} and {\bf S\ref{S:start}}. Additionally, for a fixed  $L\in \euL$ (this is interpreted as the Lagrangian of the theory), one defines $\fA_{L}$ by also quotienting by the following relation proposed by \cite{BF19} that encodes the dynamics:
\be\label{eq:SD}\tag{\bf S3}
\Scal(F)\Scal(\delta L(\ph))=\Scal(F^{\ph}+\delta L(\ph))=\Scal(\delta L(\ph))\Scal(F).
\ee
where  $F^{\ph}(\psi)\doteq F(\ph+\psi)$, $\ph,\psi\in \Ecal$ and $\delta L$ is given in Def.~\ref{df:delta:L}.

Physically, \eqref{eq:SD} is interpreted as the \textit{Schwinger-Dyson equation} on the level of local S-matrices.

\subsection{Schwinger-Dyson equation from translation symmetry}
In BV formalism in finite dimensions, the Schwinger-Dyson equation is the consequence of the translation invariance of the path integral measure. Using the formal S-matrix language, the condition {\bf\ref{eq:SD}} should be an expression of symmetry under translations in $\Ecal$.

Firstly, note that the group $(\Ecal_c,+)$ of compactly supported configurations acts on $\Ecal$ by $\sigma_\psi(\ph)=\ph+\psi$, where $\ph\in\Ecal$, $\psi\in\Ecal_c$. This induces the following map:
\begin{align*}
\alpha&:\Ecal_c\rightarrow \Aut(\euL)\\
\alpha_\psi(L)(f)[\ph]&\doteq L(f)[\sigma_\psi(\ph)]-L(f)[\ph]\,,
\end{align*}
or in shorthand notation
\[
\alpha_\psi(L)\doteq \sigma_\psi^*L-L\,.
\]
\begin{prop}
	The map $\alpha$ defined above is a 1-cocycle of $\Ecal_c$ in $ (\Aut(\euL),+)$, where the addition is the pointwise addition inherited from $\euL$.
\end{prop}
\begin{proof}
We have
\begin{multline*}
\alpha_{\psi+\chi}(L)[\ph]\doteq L[\sigma_{\psi+\chi}(\ph)]-L[\ph]=L[\ph+\psi+\chi]-L[\ph+\psi]+L[\ph+\psi]-L[\ph]\\=\alpha_{\chi}(\sigma_\psi^*L)[\ph]+\alpha_{\psi}(L)[\ph]=(\sigma_\psi^*\alpha_{\chi})(L)[\ph]+\alpha_{\psi}(L)[\ph]\,,
\end{multline*}
where in the last equation $\sigma^*$ is a map form $\Ecal$ to automorphisms of $\euL$ given by
\[
(\sigma_\psi^*(\beta))(L)\doteq \beta(\sigma_\psi^*L)\,,
\]
where $\beta\in\Aut(\euL)$. In a shorthand notation we have:
\[
\alpha_{\psi+\chi}=\sigma^*_\psi\alpha_\chi+\alpha_\psi\,,
\]
which is indeed the cocycle condition. The same argument applies if one switches $\chi$ with $\psi$, i.e.
\[
\alpha_{\psi+\chi}=\sigma^*_\chi\alpha_\psi+\alpha_\chi\,.
\]
\end{proof}
I can now re-express definition \ref{df:delta:L} as
\[
\delta L(\psi)\doteq \alpha_\psi L(f)\,, 
\]
where $\psi\in\Dcal$ and $f\equiv 1$ on $\supp \psi$. 
Hence the co-cycle $\alpha$ encodes the classical dynamics.

 The condition \eqref{eq:SD} can be re-written as
\[
\Scal(F)\Scal(\alpha_\psi L(f))=\Scal(\sigma_\psi^*F+\alpha_\psi L(f))\,.
\]
where $\psi\in\Dcal$ and $f\equiv 1$ on $\supp \psi$. Let
\[
\beta_\psi(F):= \delta L(\psi) + \sigma^*_\psi F
\]
One can define maps $\hat{\alpha}^{r/\ell}:\Ecal_c\rightarrow \Aut(\fA)$ by fixing their action on the generators, namely:
\be
\hat{\alpha}^r_\psi(\Scal(F))\doteq \Scal(\beta_\psi(F)) \Scal( \alpha_\psi L(f) )^{-1}\,\quad \hat{\alpha}^{\ell}_\psi(\Scal(F))\doteq \Scal( \alpha_\psi L(f) )^{-1}\Scal(\beta_\psi(F))\,.
\ee
To simplify the notation, I will write $\hat{\alpha}^r_\psi$ simply as $\hat{\alpha}_\psi$ and only use the superscript $r$ when distinction with $\hat{\alpha}^{\ell}_\psi$ has to be made.
\begin{prop}
	$\hat{\alpha}$ defines an action of $\Ecal_c$ on $\fA$, i.e. it is a group homomorphism from $\Ecal_c$ to an abelian subgroup of $\Aut(\fA)$:
	\[
	\hat{\alpha}_{\psi+\chi}=\hat{\alpha}_{\psi}\circ \hat{\alpha}_{\chi}=\hat{\alpha}_{\chi}\circ \hat{\alpha}_{\psi}\,.
	\]
\end{prop}
\begin{proof}
We have
\begin{multline*}
\hat{\alpha}_{\psi+\chi}(\Scal(F))=\Scal(\alpha_{\psi+\chi} L(f) + \sigma^*_{\psi+\chi} F) \Scal( \alpha_{\psi+\chi} L(f) )^{-1}=\\
\Scal( (\sigma^*_\psi\alpha_\chi+\alpha_\psi)L(f) + \sigma^*_{\psi+\chi} F) \Scal( (\sigma^*_\psi\alpha_\chi+\alpha_\psi)L(f))^{-1}=\\
\Scal( (\sigma^*_\chi\alpha_\psi+\alpha_\chi)L(f) + \sigma^*_{\psi+\psi} F) \Scal( (\sigma^*_\chi\alpha_\psi+\alpha_\chi)L(f))^{-1}\,,
\end{multline*}
where $f\equiv 1$ on the support of $\psi+\chi$.
On the other hand:
\[
\hat{\alpha}_{\psi}\circ \hat{\alpha}_{\chi}(\Scal(F))=\Scal(\alpha_\psi L(f')+\sigma^*_\psi \alpha_\chi L(f'') + \sigma^*_{\psi+\chi} F  )
\Scal(\alpha_\psi L(f')+\sigma^*_\psi(\alpha_\chi L(f'')) )^{-1}\,,
\]
where $f'\equiv 1$ on $\supp\psi$ and $f''\equiv 1$ on $\supp \chi$. Hence, taking $f$ such that $f\equiv 1$ on $\supp \chi\cup \supp \psi\supset \supp(\chi+\psi)$, we obtain:
\[
\hat{\alpha}_{\psi+\chi}=\hat{\alpha}_{\psi}\circ \hat{\alpha}_{\chi}\,.
\]
A similar argument works for $\hat{\alpha}_{\chi}\circ \hat{\alpha}_{\psi}$, hence
\[
\hat{\alpha}_{\psi+\chi}=\hat{\alpha}_{\psi}\circ \hat{\alpha}_{\chi}=\hat{\alpha}_{\chi}\circ \hat{\alpha}_{\psi}\,.
\]
\end{proof}
Using the action $\hat{\alpha}$, one can express \eqref{eq:SD} as 
\be\label{SDnew}
\hat{\alpha}^r_\psi(\Scal(F))=\Scal(F)=\hat{\alpha}^{\ell}_\psi(\Scal(F))\,,\quad \forall \psi\in \Dcal\,,
\ee
so imposing \eqref{eq:SD} amounts to quotienting $\fA$ by the action of $\Ecal_c$ (taking the space of the co-invariants), i.e. implementing \textit{translational symmetry}. This makes sense, as  \eqref{eq:SD} is the finite version of the Schwinger-Dyson equation, which, formally, is the consequence of the \textit{translation invariance} of the path integral.
\subsection{Relation to the BV perspective}

Let's consider a theory with quadratic action $L_0$. Note that, in the absence of local symmetries other than the translation symmetry,
 the infinitesimal version of $\alpha$ corresponds to the action of the classical BV operator $s_0$ and the infinitesimal version of $\beta$ is the interacting classical BV operator $s$, as demonstrated below.

In our explicit model with star product and time-ordered product, we have:
\begin{multline*}
\frac{d}{dt}\hat{\alpha}_{t\psi}(\Scal(F))\big|_{t=0}=\frac{d}{dt}(\Scal(\beta_\psi(F))\Scal(\alpha_\psi L(f))^{-1})\big|_{t=0}\\=-\tfrac{i}{\hbar}\Scal(F)\star \left<\psi,P\ph\right>+\frac{d}{dt}\Scal(\beta_\psi(F))\big|_{t=0}=\tfrac{i}{\hbar}\left( -\Scal(F)\star s_0(X_\psi)+\Scal(F)\T \frac{d}{dt}(\beta_{t\psi}F)\big|_{t=0} \right)\,.
\end{multline*}
where $X_\psi$ is a vector field defined by $X_\psi=\int \psi(x)\frac{\delta}{\delta\ph(x)}$ and 
\[
\left<\psi,P\ph\right>\doteq \frac{d}{dt} \delta L_0(t\psi)[\ph]\Big|_{t=0}
\]
is the equation of motion term $P\ph$ smeared with $\psi$. Assuming $F$ in the kernel of $s_0$, we obtain
\[
 \Scal(F)\star s_0(X_\psi)=s_0(\Scal(F)\star X_\psi)\,,
\]
The right-hand side can also be written as $s_0(\Scal(F)\T X_\psi)$, since $X_\psi$ does not depend on fields.

The map $\beta$ is the finite version of the classical interacting BV operator $s$, in the sense that
\[
\frac{d}{dt}(\beta_{t\psi}F)\big|_{t=0}= \left<\psi,P\ph\right>+ \partial_{X_\psi}F=s(X_\psi)\,,
\]
where $s$ is the classical BV operator for the theory with the interaction $S_0+F$.  Using the antibracket notation, we can write the above formula also as $\{ X_\psi,S_0+F\}$.

Putting all these together, we obsere that the infinitesimal version of equation \eqref{SDnew} is
\[
\frac{d}{dt}\hat{\alpha}_{t\psi}(\Scal(F))\big|_{t=0}=-s_0(\Scal(F)\T X_\psi)+\Scal(F)\T s( X_\psi)=0\,.
\]
Comparing this with \eqref{MWI:inf}, the infinitesimal version of {\mwi}, we notice that in this case the term corresponding to the BV Laplacian $\Lap$ vanishes, since $X_\psi$ does not depend on fields. We can therefore write the infinitesimal version of $\hat{\alpha}$ in a conceptually more appropriate form:
\[
\frac{d}{dt}\hat{\alpha}_\psi(\Scal(F))\big|_{t=0}=\Scal(F)\T \hat{s}( X_\psi)-s_0(\Scal(F)\T X_\psi)\,.
\]
\begin{rem}
From the conceptual viewpoint, it seems that the Schwinger-Dyson equation relates the classical BV operator to the quantum one, so it relates two cohomology theories. It is not, by itself, a relation that corresponds to computing any of these cohomologies on their own. Compare this with the infinitesimal version {\mwi} \eqref{MWI:inf}, which plays the same role.
\end{rem}
\subsection{Action by diffeomorphisms}
The constructios of the previous sections should generalize from translations to arbitrary compactly supported diffeomorphisms of $\Ecal$ and we restrict ourselves to the diffeomorphisms that arise as exponentiated local, compactly supported vector fields $X\in\Vcal$. I denote this space by $\Diff_{\loc}$
Firstly, we extend $\alpha$ to $$\alpha:\Diff_{\loc}(\Ecal)\rightarrow \Aut(\euL)$$ by setting
\[
\alpha_g(L)(f)[\ph]\doteq L(f)[g(\ph)]-L(f)[\ph]\,.
\]
Similarly, define 
\[
\beta_g(F)\equiv \alpha_g L(f) + \sigma^*_g F\,,
\]
where $f\equiv 1$ on $\supp g$ and $\sigma^*_g F(\ph)\doteq F(g(\ph))$. More generally, we can consider the ``anomalous'' version of the action $\beta$:
\be\label{AWI:beta}
\beta_g(F)\doteq \alpha_g L(f) + \sigma^*_g F+\Acal_F(g)\,,
\ee
where again  $f\equiv 1$  and $\Acal_F(X)$ is the ``finite version'' of the  renormalized BV Laplacian \eqref{BVLap}, aka ``the anomaly term,'' meaning that
\[
\frac{d}{dt} \Acal_F(e^{tX})\big|_{t=0}=\Lap_F(X)
\]
In the path integral language, $\Acal_F(g)$ should be the logarithm of the Jacobian of the configuration space transformation $g\in G$. If no renormalization was needed, it would not depend on $F$. The interpretation of $\Acal_F(X)$ as the log Jacobian is consistent with the infinitesimal formulation, since the derivative of the Jacobian of $g$ gives the divergence of the vector filed $X$, generating $g$. 

In perturbation theory, the Master Ward Identity is the generalization of the Schwinger-Dyson equation, where translations are replaced by general compactly supported diffeomorphisms. Similarly, the \textit{unitary Master Ward Identity} should be an appropriate generalization of the \textit{unitary} Schwinger-Dyson equation \eqref{SDnew}. Note, however, that \eqref{SDnew} cannot hold for general $g$ (more general than the translations), if we keep the same definition of $\hat{\alpha}^{r/\ell}$. This is because, $\Scal(\alpha_g L(f))$ is no longer central, so multiplying from the right cannot be the same as multiplying from the left in \eqref{SDnew}.

A proposal for the on-shell formulation of the unitary Master Ward Identity will be presented in the upcoming paper by Brunetti, D\"utsch, Fredenhagen and myself \cite{Future}. In that paper we will also formulate the unitary quantum version of Noether's theorem. 

As for the off-shell version, analogous to the BV formalism, the following geometrical perspective might give a hint as to what is the correct formulation.

\subsection{Geometrical picture}

\subsubsection{Vector fields}
Let $\Ecal^*$ be the space of smooth sections of the dual bundle and let $\Ecal^!\equiv \Ecal^*\otimes{\rm Dens}$ (tensoring with densities). We can identify vector fields with functions on $\bar{\Ecal}\equiv\Ecal\times\Ecal^!\subset T^*\Ecal$ that are linear in the second argument. We will use the notation $X=\int X(x)\frac{\delta}{\delta\ph(x)}$ and identify $q$ with $\frac{\delta}{\delta \ph}$ so that 
\[
X(\ph,q)=\int  X(x)[\ph]q(x)\,,
\]
where $q\in\Ecal^!$.


\subsubsection{Products}
Let $\bar{\Fcal}$ denote the space of smooth functionals on $\bar{\Ecal}$ and let $\bar\Fcal_\loc$ be it's subspace consisting of the local ones. 
The following products on $\bar{\Fcal}$ will be used to encode how diffeomorphisms $e^X$ are acting on functions on the base:
\begin{align*}
	F\astunder G&=m\circ e^{-i\Lap_\otimes}(F\otimes G)\,,\\
	F\,\astover \,G&=m\circ e^{i\Lap_\otimes^T}(F\otimes G)\,,
\end{align*}
where 
\[
\Lap_\otimes\doteq \int \frac{\delta}{\delta q(x)}\otimes  \frac{\delta}{\delta \ph(x)}\,.
\]
The two products are related by complex conjugation in the sense that
\[
\overline{F\astunder G}=\bar{G}{\astover}\bar{F}\,.
\]
A useful identity involves the Jacobian of the diffeomorphism $ e^{X}$. Denote
\[
\Lap\doteq \int \frac{\delta^2}{\delta q(x)\delta \ph(x)}\,.
\]
Using this notation, $\Lap(X)$ is the divergence of the vector field $X\in\euV$. 	The exponential of $\Lap$ also allows one to relate the products:
\[
		e^{i\Lap} (X\astunder Y)=(e^{i\Lap}X)\astover (e^{i\Lap}Y)
\]
Hence
\[e^{-i\Lap}\circ \exp_{\astover}(i X)=\exp_{\astunder}(e^{-i\Lap}(i X))=e_{\astunder}^{iX+\Lap X}\,.
\]
We also have:
\[		e_{\astunder}^{iX}{\astunder}F{\astunder}e_{\astunder}^{-iX}= \sigma_g^*F\,,
\]
as well as:
	\[
	e_{\astunder}^{iX+\Lap(X)}\astunder F\astunder e_{\astunder}^{-iX}=\Jac( g)\, \sigma_g^*F \,,
	\]
where $g=e^X$ is the diffeomorphism given by exponentiating $X$. These equations allow one to encode the action of symmetries on functionals using the above products on $\overline{\Fcal}$.

Let $L_0'$ and $L_0''$ denote the first and second derivatives of a quadratic Lagrangian $L_0$, in the sense that
\[
\left<L_0',\psi\right> =\frac{d}{dt} \delta L_0(t\psi)\big|_{t=0}\,,
\]
for $\psi\in \Ecal_c$ and
\[
\left<L_0'',\psi_1\otimes \psi_2\right> =\frac{d^2}{dtds} \delta L(t\psi_1+s\psi_2)\big|_{t=s=0}\,,
\]
where $\psi_1, \psi_2 \in \Ecal_c$, but due to locality, one of the arguments could be non-compactly supported, so $L''$ induces a differential operator $P:\Ecal\rightarrow\Ecal$.

Let $g=e^X$, one can then write
\[
\alpha_gL=\left<L',\frac{\delta}{\delta q} e_{\astunder}^{X}\big|_{q=0}\right>+\frac{1}{2} \left<L'',\frac{\delta^2}{\delta q^2} e_{\astunder}^{X}\big|_{q=0}\right>=\left(L+\left<L',\frac{\delta}{\delta q}\right>+\frac{1}{2} \left<L'',\frac{\delta^2}{\delta q^2} \right>\right)e_{\astunder}^{X}\big|_{q=0}
\]
 Using the short-hand $\frac{\delta}{\delta q}\equiv \delta_q$ and defining
\[
\delta L(\delta_q)\doteq \left<L',\frac{\delta}{\delta q}\right>+\frac{1}{2} \left<L'',\frac{\delta^2}{\delta q^2} \right>\,,
\]
one obtains for $f\equiv 1$ on $\supp X$:
\[
\alpha_gL(f)= \delta L(\delta_q) e_{\astunder}^{X}\big|_{q=0}
\]
The exponential of $\delta L(\delta_q)$ can be thought of as the finite version of $\delta_0$ and hence one can define a map
\[
\tilde{s}_0=\TT^{-1}\circ 	e^{-i\delta L(i\delta_q)}\circ \TT\,,
\]
in analogy to equation \eqref{quantumBV} in the BV formalism. Using the Baker–Campbell–Hausdorff formula one finds that 
\[
e^{-i\delta L(i\delta_q)}\circ \TT=	\TT\circ e^{i\delta L(-i\delta_q)}\circ e^{-i\Lap}
\]
which leads to the following heuristic proposal for the finite off-shell anomalous MWI:
\be\label{MWI2}
e^{-i\delta L(i\delta_q)}\Tcal(e_{{\astover}}^{iX}\astover e^{iF})=\Tcal\left(\Jac(e^X)e^{i(\sigma_g^*F+\alpha_gL (X))}\, e_{\astunder}^{iX}\right)=\Tcal\left(e^{i(\sigma_g^*F+\alpha_gL(X)-i\Acal_F(g))}\, e_{\astunder}^{iX}\right)\,,
\ee
which holds identically for regular functionals. Barring the IR problems we also have
\[
e^{iL}{\astover}\Tcal(e_{{\astover}}^{iX}e^{iF}){\astover}	e^{-iL}=e^{-i\delta L(i\delta_q)}\circ \TT (e_{{\astover}}^{iX}e^{iF})
\]

One can easily imagine that exponentials appearing in the above equation could be identified with some generalization of $S$-matrices. However, as it stands, these are not unitary with respect to the star product present in this setting. A similar situation appears if one takes time-ordered exponentials of regular, but not local, functionals. Hence, it is likely that unitarity can be restored as one reverts to local functionals, but then renormalization has to be performed. This issue and other questions posed in this review will be further addressed in my future works.

\appendix
\section{Proof of theorem \ref{MollerBV}}
The proof relies on two lemmas.
 \begin{lemma}
 	\[
 	r_{V}^{-1}(\{X,S_{00}\})=\{r_{V}^{-1}X,S_{00}+V\}- \int \frac{\delta_r X}{\delta\ph^\gamma(y)}({\mathtt r}^{-1}_{ V}(\ph))\Delta^{\rm R}(y,z)^{\gamma\beta}\frac{\delta_l}{\delta\ph^\beta(z)}\left(\frac{1}{2}\{S_{00}+V,S_{00}+V\}\right)
 	\]
 \end{lemma}
 \begin{proof}
By definition of $\{.,.\}$ and $r_V^{-1}$ we have:
 	\begin{align}
 	r_{V}^{-1}(\{X,S_{00}\})&= -\int \frac{\delta_r X}{\delta\ph^\ddagger_\alpha(x)}({\mathtt r}^{-1}_{ V}(\ph))\frac{\delta_l S_{00}}{\delta\ph^\al(x)}({\mathtt r}^{-1}_{V}(\ph))\nonumber
 	\\
 	&=-\int \frac{\delta_r X}{\delta\ph^\ddagger_\alpha(x)}({\mathtt r}^{-1}_{V}(\ph))\left(\frac{\delta_l S_{00}}{\delta\ph^\al(x)}(\ph)+\frac{\delta_l V}{\delta\ph^\al(x)}(\ph)\right)\,,\label{eq:star}
 	\end{align}
 	where we used the fact that $\frac{\delta_l S_{00}}{\delta\ph^\al(x)}(\ph)=P_{\alpha\beta}(x)(\ph^\beta(x))$,	so
 	\[
 	\frac{\delta_l S_{00}}{\delta\ph^\al(x)}({\mathtt r}^{-1}_{V}(\ph))(\ph)=\frac{\delta_l S_{00}}{\delta\ph^\al(x)}(\ph)+\frac{\delta_l V}{\delta\ph^\al(x)}(\ph)\,,
 	\]
 	as $P\circ \Delta^{\rm R}(x,y)=\delta(x-y)$. Next, we notice that applying the chain rule:
 	\be\label{eq:heart}
 	\frac{\delta_r}{\delta\ph_\alpha^\ddagger(x)}(r_V^{-1}(X))= \frac{\delta_r X}{\delta\ph^\ddagger_\alpha(x)}({\mathtt r}^{-1}_{ V}(\ph))+\int\frac{\delta_r X}{\delta\ph^\gamma(y)}({\mathtt r}^{-1}_{ V}(\ph))\Delta^{\rm R}(y,z)\frac{\delta_l}{\delta\ph^\beta(z)}\frac{\delta_r V}{\delta\ph^\ddagger_\alpha(x)}(\ph)\,.
 	\ee
 	We solve the above equation for $\frac{\delta_r X}{\delta\ph^\ddagger_\alpha(x)}({\mathtt r}^{-1}_{ V}(\ph))$ and insert it into the first term of \eqref{eq:star} to obtain
 	\begin{multline}\label{eq:twostar}
 	-\int \frac{\delta_r X}{\delta\ph^\ddagger_\alpha(x)}({\mathtt r}^{-1}_{V}(\ph))\frac{\delta_l S_{00}}{\delta\ph^\al(x)}(\ph)\\=\int 
 	\left(-\frac{\delta_r}{\delta\ph_\alpha^\ddagger(x)}(r_V^{-1}(X))+ \int\frac{\delta_r X}{\delta\ph^\gamma(y)}({\mathtt r}^{-1}_{ V}(\ph))\Delta^{\rm R}(y,z)\frac{\delta_l}{\delta\ph^\beta_g(z)}\frac{\delta_r V}{\delta\ph^\ddagger_\alpha(x)}(\ph)\right)
 	\frac{\delta_l S_{00}}{\delta\ph^\al(x)}(\ph)
 	\end{multline}
 	The first term in \eqref{eq:twostar} is what we eventually want to get on the right-hand side of the identity we are trying to prove and the second term can be rewritten using the Leibniz rule into:
 	\be\label{eq:dagger}
 	\int \frac{\delta_r X}{\delta\ph^\gamma(y)}({\mathtt r}^{-1}_{ V}(\ph))\Delta^{\rm R}(y,z)^{\gamma\beta}\left(\frac{\delta_l}{\delta\ph^\beta(z)}\left(\frac{\delta_r V}{\delta\ph^\ddagger_\alpha(x)}
 	\frac{\delta_l S_{00}}{\delta\ph^\al(x)}\right)-(-1)^{|\beta|(|\alpha|+1)}\frac{\delta_r V}{\delta\ph^\ddagger_\alpha(x)} \frac{\delta_l}{\delta\ph^\beta(z)}\frac{\delta_l S_{00}}{\delta\ph^\al(x)}\right)
 	\ee
 	Now we use the fact that $\frac{\delta_l S_{00}}{\delta\ph^\al(x)}=(-1)^{|\alpha|}\frac{\delta_r S_{00}}{\delta\ph^\al(x)}$ and that composing $\Delta^{\rm R}$ with $P$ gives identity. This allows to rewrite \eqref{eq:dagger} as
 	\be\label{eq:doubledagger}
 	\int \frac{\delta_r X}{\delta\ph^\gamma(y)}({\mathtt r}^{-1}_{ V}(\ph))\Delta^{\rm R}(y,z)^{\gamma\beta}\frac{\delta_l}{\delta\ph^\beta(z)}\left(\frac{\delta_r V}{\delta\ph^\ddagger_\alpha(x)}
 	\frac{\delta_l S_{00}}{\delta\ph^\al(x)}\right)- (-1)^{|\alpha|} \int \frac{\delta_r X}{\delta\ph^\alpha(x)}({\mathtt r}^{-1}_{ V}(\ph)) \frac{\delta_r V}{\delta\ph^\ddagger_\alpha(x)}
 	\ee
 	There are two terms in \eqref{eq:doubledagger}: the first one is in the form we want, but we need to work a bit more on the second one.
 	First, we notice that
 	\be\label{eq:threestar}
 	-(-1)^{|\alpha|} \int \frac{\delta_r X}{\delta\ph^\alpha(x)}({\mathtt r}^{-1}_{ V}(\ph)) \frac{\delta_r V}{\delta\ph^\ddagger_\alpha(x)}=\int \frac{\delta_r X}{\delta\ph^\alpha(x)}({\mathtt r}^{-1}_{ V}(\ph)) \frac{\delta_l V}{\delta\ph^\ddagger_\alpha(x)}
 	\ee
and using a calculation analogous to \eqref{eq:heart}, we also get
 	\[\frac{\delta_r X}{\delta\ph^\alpha(x)}({\mathtt r}^{-1}_{ V}(\ph))=
 	\frac{\delta_r}{\delta\ph^\alpha(x)}(r_V^{-1}(X))-\int\frac{\delta_r X}{\delta\ph^\gamma(y)}({\mathtt r}^{-1}_{ V}(\ph))\Delta^{\rm R}(y,z)\frac{\delta_l}{\delta\ph^\beta(z)}\frac{\delta_r V}{\delta\ph^\alpha(x)}(\ph)\,,
 	\]
 	which we then insert on the right-hand side of \eqref{eq:threestar}. Plugging the result back into \eqref{eq:doubledagger} and then into \eqref{eq:twostar}, we obtain the final expression:
 	\begin{align*}
 	r_{V}^{-1}(\{X,S_{00}\})
 	&=\{r_V^{-1}(X),S_{00}+V\}+ \int \frac{\delta_r X}{\delta\ph^\gamma(y)}({\mathtt r}^{-1}_{ V}(\ph))\Delta^{\rm R}(y,z)^{\gamma\beta}\frac{\delta_l}{\delta\ph^\beta(z)}\left(-\{S_{00},V\}-\frac{1}{2}\{V,V\}\right)\\
 	&=\{r_V^{-1}(X),S_{00}+V\}- \int \frac{\delta_r X}{\delta\ph^\gamma(y)}({\mathtt r}^{-1}_{ V}(\ph))\Delta^{\rm R}(y,z)^{\gamma\beta}\frac{\delta_l}{\delta\ph^\beta(z)}\left(\frac{1}{2}\{S_{00}+V,S_{00}+V\}\right)\,.
 	\end{align*}

 \end{proof}
 
 \begin{lemma}
 	\[
 	r_{V}^{-1}(\{X,\theta_0\})=\{r_{V}^{-1}X,
 	\theta_0\}- \int \frac{\delta_r X}{\delta\ph^\gamma(y)}({\mathtt r}^{-1}_{ V}(\ph))\Delta^{\rm R}(y,z)^{\gamma\beta}\frac{\delta_l}{\delta\ph^\beta(z)}\left(\{V,\theta_0\}_g\right)
 	\]
 \end{lemma}
 \begin{proof}
 	We have
 	\be\label{eq:star2}
 	r_{V}^{-1}(\{X,\theta_0\})=\int \frac{\delta_r X}{\delta\ph^\alpha(x)}({\mathtt r}^{-1}_{ V}(\ph))\frac{\delta_l \theta_0}{\delta\ph^\ddagger_\al(x)}({\mathtt r}^{-1}_{V}(\ph)) -\int \frac{\delta_r X}{\delta\ph^\ddagger_\alpha(x)}({\mathtt r}^{-1}_{ V}(\ph))\frac{\delta_l \theta_0}{\delta\ph^\al(x)}({\mathtt r}^{-1}_{V}(\ph))
 	\ee
 	Note that $\frac{\delta \theta_0}{\delta\ph^\al(x)}({\mathtt r}^{-1}_{V}(\ph))=\frac{\delta \theta_0}{\delta\ph^\al(x)}(\ph)$,  as $\theta_0$ is linear in fields and that
 	\[
 	\frac{\delta_l \theta_0}{\delta\ph^\ddagger_\al(x)}({\mathtt r}^{-1}_{V}(\ph))
 	=
 	K^\alpha_{\ \sigma}\left(\ph^\sigma(x)+  \int_M\Delta^{\mathrm{R}}(x,y)^{\sigma\beta}\, \frac{\delta_l V}{\delta\ph^\beta_g(y)}(\ph)\right)=K^\alpha_{\ \sigma} \ph^\sigma(x)+  {\int_MK^\alpha_{\ \sigma}(x)\Delta^{\mathrm{R}}(x,y)^{\sigma\beta}\, \frac{\delta_lV}{\delta\ph^\beta_g(y)}(\ph)}
 	\]
 	%
 	So the first term in \eqref{eq:star2} becomes:
 	\begin{multline}\label{eq:dagger2}
 	\int \frac{\delta_r X}{\delta\ph^\alpha(x)}({\mathtt r}^{-1}_{ V}(\ph))\frac{\delta_l \theta_0}{\delta\ph^\ddagger_\al(x)}({\mathtt r}^{-1}_{V}(\ph))\\=
 	\int \frac{\delta_r X}{\delta\ph^\alpha(x)}({\mathtt r}^{-1}_{ V}(\ph))K^\alpha_{\ \sigma} \ph^\sigma(x)+  {\int \frac{\delta_r X}{\delta\ph^\alpha(x)}({\mathtt r}^{-1}_{ V}(\ph)) K^\alpha_{\ \sigma}(x)\Delta^{\mathrm{R}}(x,y)^{\sigma\beta}\, \frac{\delta_lV}{\delta\ph^\beta_g(y)}(\ph)}\,.
 	\end{multline}
 	Now we use the fact that
 	\[
 	\frac{\delta_r X}{\delta\ph^\alpha(x)}({\mathtt r}^{-1}_{ V}(\ph))= 	
 	\frac{\delta_r}{\delta\ph^\alpha(x)}(r_V^{-1}(X))
 	-\int\frac{\delta_r X}{\delta\ph^\gamma(y)}({\mathtt r}^{-1}_{ V}(\ph))\Delta^{\rm R}(y,z)\frac{\delta_l}{\delta\ph^\beta(z)}\frac{\delta_r V}{\delta\ph^\alpha(x)}(\ph)
 	\]
 	to rewrite  the first term in \eqref{eq:dagger2} as
 	\begin{multline}\label{eq:doubledagger2}
 	\int\frac{\delta_r X}{\delta\ph^\alpha(x)}({\mathtt r}^{-1}_{ V}(\ph)) K^{\alpha}_{\ \sigma}\ph^\sigma(x)= 	\int\frac{\delta_r}{\delta\ph^\alpha(x)}(r_V^{-1}(X))K^{\alpha}_{\ \sigma}\ph^\sigma(x)\\
 	-\int \frac{\delta_r X}{\delta\ph^\gamma(y)}({\mathtt r}^{-1}_{ V}(\ph))\Delta^{\rm R}(y,z)^{\gamma\beta}\frac{\delta_l}{\delta\ph^\beta(z)}\frac{\delta_r V}{\delta\ph^\alpha(x)}(\ph)K^{\alpha}_{\ \sigma}\ph^\sigma(x)\,,
 	\end{multline}
 	The first term in \eqref{eq:doubledagger2} is what we want and we rewrite the second term, using the Leibniz rule, as:
 	\begin{multline}\label{eq:spades}
 	-\int\frac{\delta_r X}{\delta\ph^\gamma(y)}({\mathtt r}^{-1}_{ V}(\ph))\Delta^{\rm R}(y,z)^{\gamma\beta}\frac{\delta_l}{\delta\ph^\beta(z)}\left(\frac{\delta_r V}{\delta\ph^\alpha(x)}(\ph)\frac{\delta_l \theta_0}{\delta\ph^\ddagger_\al(x)}(\ph)\right)\\
 	+\int{\frac{\delta_r X}{\delta\ph^\gamma(y)}({\mathtt r}^{-1}_{ V}(\ph))\Delta^{\rm R}(y,z)^{\gamma\beta}\frac{\delta_r V}{\delta\ph^\alpha(x)}(\ph)\frac{\delta_l}{\delta\ph^\beta(z)}\frac{\delta_l \theta_0}{\delta\ph^\ddagger_\al(x)}(\ph)}
 	\end{multline}
 	Now we put \eqref{eq:spades} back into \eqref{eq:doubledagger2}, use the fact that $\frac{\delta_l}{\delta\ph^\beta(z)}\frac{\delta_l \theta_0}{\delta\ph^\ddagger_\al(x)}(\ph)=K^\alpha_{\ \beta}(x)\delta(x-z)$ and insert the resulting form of \eqref{eq:doubledagger2} into \eqref{eq:dagger2}, to obtain:
 	\begin{multline}\label{eq:clubs}
 	\int\frac{\delta_r X}{\delta\ph^\alpha(x)}({\mathtt r}^{-1}_{ V}(\ph))	\frac{\delta_l \theta_0}{\delta\ph^\ddagger_\al(x)}({\mathtt r}^{-1}_{V}(\ph))=\int \frac{\delta_r}{\delta\ph^\alpha(x)}(r_V^{-1}(X))K^{\alpha}_{\ \sigma}\ph^\sigma(x)\\
 	-\int \frac{\delta_r X}{\delta\ph^\gamma(y)}({\mathtt r}^{-1}_{ V}(\ph))\Delta^{\rm R}(y,z)\frac{\delta_l}{\delta\ph^\beta(z)}\left(\frac{\delta_r V}{\delta\ph^\alpha(x)}(\ph)\frac{\delta_l \theta_0}{\delta\ph^\ddagger_\al(x)}(\ph)\right)\\
 	-\int {\frac{\delta_r X}{\delta\ph^\gamma(y)}({\mathtt r}^{-1}_{ V}(\ph))
 		K^\alpha_{\ \beta}(z)	
 		\Delta^{\rm R}(y,z)^{\gamma\beta}\frac{\delta_r V}{\delta\ph^\alpha(x)}(\ph)}
 	+{\int\frac{\delta_r X}{\delta\ph^\alpha(x)}({\mathtt r}^{-1}_{ V}(\ph))K^\alpha_{\ \sigma}(x)\Delta^{\mathrm{R}}(x,y)^{\sigma\beta}\, \frac{\delta_lV}{\delta\ph^\beta_g(y)}(\ph)}
 	\end{multline}
 	The last two terms in \eqref{eq:clubs} combine to
 	\[
 	\int\frac{\delta_r X}{\delta\ph^\gamma(y)}({\mathtt r}^{-1}_{ V}(\ph))
 	\left(-	K^\alpha_{\ \beta}(z)	
 	\Delta^{\rm R}(y,z)^{\gamma\beta}
 	+(-1)^{|\alpha|}K^\gamma_{\ \beta}(y)\Delta^{\rm R}(y,z)^{\beta\alpha}	
 	\right)\frac{\delta_r V}{\delta\ph^\alpha(z)}(\ph)\,,
 	\]
 	which vanishes due to \eqref{const:cond}.
 	Next, using \eqref{eq:heart}, we compute the second term in \eqref{eq:star2}, namely:
 	\begin{multline}\label{eq:diamonds}
 	-\int \frac{\delta_r X}{\delta\ph^\ddagger_\alpha(x)}({\mathtt r}^{-1}_{ V}(\ph))\frac{\delta_l \theta_0}{\delta\ph^\al(x)}({\mathtt r}^{-1}_{V}(\ph))\\
 	=\int \left(-\frac{\delta_r}{\delta\ph_\alpha^\ddagger(x)}(r_V^{-1}(X))+ \frac{\delta_r X}{\delta\ph^\gamma(y)}({\mathtt r}^{-1}_{ V}(\ph))\Delta^{\rm R}(y,z)\frac{\delta_l}{\delta\ph^\beta_g(z)}\frac{\delta_r V}{\delta\ph^\ddagger_\alpha(x)}(\ph)\right)\frac{\delta_l \theta_0}{\delta\ph^\al(x)}(\ph)\\
 	=-\int \frac{\delta_r}{\delta\ph_\alpha^\ddagger(x)}(r_V^{-1}(X))\frac{\delta_l \theta_0}{\delta\ph^\al(x)}(\ph)
 	+\int \frac{\delta_r X}{\delta\ph^\gamma(y)}({\mathtt r}^{-1}_{ V}(\ph))\Delta^{\rm R}(y,z)\frac{\delta_l}{\delta\ph^\beta_g(z)}\left(\frac{\delta_r V}{\delta\ph^\ddagger_\alpha(x)}(\ph)\frac{\delta_l \theta_0}{\delta\ph^\al(x)}(\ph)\right)
 	\end{multline}
 	Inserting \eqref{eq:clubs} and \eqref{eq:diamonds} back into \eqref{eq:star2} we obtain finally:
 	\begin{multline*}
 	r_{V}^{-1}(\{X,\theta_0\})=\{r_{V}^{-1}(X),\theta_0\}\\
 	-\int \frac{\delta_r X}{\delta\ph^\gamma(y)}({\mathtt r}^{-1}_{ V}(\ph))\Delta^{\rm R}(y,z)\frac{\delta_l}{\delta\ph^\beta(z)}\left(\frac{\delta_r V}{\delta\ph^\alpha(x)}(\ph)\frac{\delta_l \theta_0}{\delta\ph^\ddagger_\al(x)}(\ph)-\frac{\delta_r V}{\delta\ph^\ddagger_\alpha(x)}(\ph)\frac{\delta_l \theta_0}{\delta\ph^\al(x)}(\ph)\right)
 	\end{multline*}
 \end{proof}
 Now we are ready to prove our main result.
 \begin{proof}(of Theorem \ref{MollerBV})
 	Since $S_0=\S_{00}+\theta_0$, the lemmas imply that
 	\[
 	r_{V}^{-1}(\{X,S_{0}\})=\{r_{V}^{-1}X,S_{0}+V\}\\- \int \frac{\delta_r X}{\delta\ph^\gamma(y)}({\mathtt r}^{-1}_{ V}(\ph))\Delta^{\rm R}(y,z)^{\gamma\beta}\frac{\delta_l}{\delta\ph^\beta(z)}\big(\{S_{0},V\}+\tfrac{1}{2}\{V,V\}\big)
 	\]
 	The last term in the brackets can be rewritten as
 	$	\frac{1}{2}(\{S,S\}-\{S_0,S_0\})$, 	so is a difference between the {\cme} for the full action and the {\cme} of the linearized theory.
 \end{proof}
 {
\bibliographystyle{amsalpha}
\bibliography{References}
}
\end{document}